\documentclass{llncs}

\usepackage{xspace}
\usepackage[boxed]{algorithm2e}
\usepackage{graphicx}
\usepackage{url}
\usepackage{wrapfig}
\usepackage{subfig}

\usepackage{amsmath}
\usepackage{amssymb}

\usepackage[T1]{fontenc}
\usepackage{slantsc}

\newcommand{\secspace}{\vspace{-2mm}}
\newcommand{\subsecspacea}{\vspace{-3mm}}
\newcommand{\subsecspaceb}{\vspace{-1mm}}

\usepackage[margins]{trackchanges} 

\newcommand{\new}{}

\newcommand{\CGAL}{\textsc{Cgal}\xspace}

\newcommand{\CC}{C\raise.08ex\hbox{\tt ++}\xspace}

\newcommand{\NN}{nearest-neighbor\xspace}
\newcommand{\PL}{point location\xspace}
\newcommand{\PPL}{planar point location\xspace}
\newcommand{\lqpl}{$\cal{L}$\xspace} 
\newcommand{\depth}{$\cal{D}$\xspace} 
\newcommand{\NAME}{ENNRIC\xspace}

\newcommand{\N}{\ensuremath{\mathbb{N}}}

  \renewenvironment{thebibliography}[1]{%
    \begin{oldthebibliography}{#1}%
      \footnotesize
      \setlength{\parskip}{0ex}%
      \setlength{\itemsep}{0ex}%
  }%
  {%
    \end{oldthebibliography}%
  }

\newcommand{\ignore}[1]{}
\newcommand{\first}[2]{#1}
\newcommand{\second}[2]{#2}

\def\marrow{\marginpar[\hfill$\longrightarrow$]{$\longleftarrow$}}

\second{
\def\dan#1{\textcolor{red}{\textsc{Danny says: }{\marrow\sf #1}}}
\def\mike#1{\textcolor{blue}{\textsc{Mike says: }{\marrow\sf #1}}}
\def\michal#1{\textcolor{cyan}{\textsc{Michal says: }{\marrow\sf #1}}}
}{
\newcommand{\dan}[1]{}
\newcommand{\mike}[1]{}
\newcommand{\michal}[1]{}
}

\begin{document}


\mainmatter

\title{Improved  Implementation of Point Location in General Two-Dimensional Subdivisions%
\thanks{This work has been supported in part by the 7th Framework
Programme for Research of the European Commission, under
FET-Open grant number 255827 (CGL---Computational Geometry
Learning), by the Israel Science Foundation (grant no.
1102/11), and by the Hermann Minkowski--Minerva Center for
Geometry at Tel Aviv University.}}
\author{Michael Hemmer \and Michal Kleinbort \and Dan Halperin}

\tocauthor{%
Michael Hemmer (Tel-Aviv University) ,
Michal Kleinbort (Tel-Aviv University),
Dan Halperin (Tel-Aviv University)}

\institute{
Tel-Aviv University, Israel \\
}

\maketitle

\begin{abstract}
We present a major revamp of the point-location data structure for
general two-dimensional subdivisions via randomized incremental
construction,
implemented in  \CGAL, the Computational Geometry Algorithms Library.
We can now guarantee that the constructed
directed acyclic graph~$\cal G$ is of linear size and
provides logarithmic query time.
Via the construction of the Voronoi diagram for a given point set S of size n,
this also enables nearest-neighbor queries in guaranteed $O(\log n)$ time.
%
%
Another major innovation is the support of general unbounded
subdivisions as well as subdivisions of two-dimensional
parametric surfaces such as spheres, tori, cylinders.
The implementation is exact, complete, and
general, i.e., it can also handle non-linear subdivisions.
Like the previous version, the data structure supports
modifications of the subdivision,
such as insertions and deletions of edges,
after the initial preprocessing.
A major challenge is to retain the expected~$O(n \log n)$
preprocessing time while providing the above (deterministic)
space and query-time guarantees.
We describe an efficient preprocessing algorithm, which explicitly 
verifies the length~\lqpl of the longest query path in $O(n \log n)$ time.
However, instead of using~\lqpl, our implementation is based
on the depth~\depth of~$\cal G$.
Although we prove that the worst case ratio
of~\depth and~\lqpl is~$\Theta(n/\log n)$,
we conjecture, based on our experimental results,
that this solution achieves expected~$O(n \log n)$ preprocessing time.
%
\end{abstract}


\mike{Find a name to address the RIC. RIC is too general as there are many more algorithms that are randomized incremental constructions. }
\mike{Find a name for our NN implementation via Voronoi Diagram. Danny proposed ENNRICH. However, the H mean hierarchy and is not intuitive to me. RIC is debatable anyway, see above. ideas? }

\secspace\section{Introduction} \secspace
\label{sec:introduction}

Birn et al.~\cite{SANDERS2010_simp_fast_nn}  presented
a structure for planar \NN queries, based on Delaunay triangulations,
named Full Delaunay Hierarchies (FDH).
The FDH is a very simple, and thus light, data structure that is also very
easy to construct. It outperforms many other methods in several scenarios,
but it does not have a worst-case optimal behavior.
However, it is claimed~\cite{SANDERS2010_simp_fast_nn} that methods that do have this behavior are too
cumbersome to implement and thus not available. We got challenged by this claim.

In this article we present an improved version of \CGAL's planar point location that
implements the famous incremental construction (RIC) algorithm
as introduced by
Mulmuley~\cite{Mulmuley1990_fast_planar_part_alg} and
Seidel~\cite{Seidel1991_simp_fast_inc_rand_td}.
The algorithm constructs 
a linear size data structure that guarantees a logarithmic query time.
It enables \NN queries in guaranteed $O(\log n)$ time via
\PPL in the Voronoi Diagram of the input points.
In Section~\ref{sec:nn}  we compare our revised implementation for point location,
applied to nearest neighbor search, against the FDH.
Naturally, this is only a byproduct of our efforts
as \PPL is a very fundamental problem in
Computational Geometry.
It has numerous applications in a variety of domains including
computer graphics, motion planning, computer aided design (CAD) and
geographic information systems (GIS).

\textbf{Previous Work:}
Most solutions can only provide an expected query time of $O(\log{n})$
but cannot guarantee it, in particular, those that only require $O(n)$ space.
Some may be restricted to static scenes that do not change, while others
can only support linear geometry.


Triangulation-based \PL methods, such as the approaches by
Kirkpatrick~\cite{Kirkpatrick1983_opt_stch_subdivisions} and
Devillers~\cite{Devillers2002_del_hierarchy} combine a logarithmic
hierarchy with some walk strategy.
Both
require only linear space and Kirkpatrick can even
guarantee logarithmic query time.
However, both are restricted to linear geometry, since they build on
a triangulation of the actual input.

Many methods can be summarized under the model of the trapezoidal
search graph as pointed out by Seidel and Adamy~\cite{SeidelA2000_wc_query_comp}.
Conceptually, the initial subdivision is further subdivided into trapezoids
by emitting vertical rays (in both directions) at every endpoint of the input,
which is the fundamental search structure.
In principal, all these solutions can be generalized to support input curves that are
decomposable into a finite number of $x$-monotone pieces.

The \emph{slabs method} of Dobkin and
Lipton~\cite{DobkinL1976_multidim_srch_prblm} is one of the earliest examples.
%
%
Every endpoint induces a vertical wall giving rise to $2n+1$ vertical slabs.
A point location is performed by a binary search to locate the correct slab
and another search within the slab in $O(\log n)$ time.
Preparata~\cite{Preparata1981_trpz_graph_pl} introduced a method that
avoids the decomposition into $n+1$ slabs reducing the required space from
$O(n^2)$ to $O(n\log{n})$.
Sarnak and Tarjan~\cite{SarnakT1986_pl_perst_trees} went back to the
slabs of Dobkin and Lipton and added the idea of persistent data structures,
which reduced the space consumption to $O(n)$.
Another example for this model is the separating chains method of
Lee and Preparata~\cite{LeeP1976_sep_chains_pl}.
Combining it with fractional cascading,
Edelsbrunner et~al.~\cite{EdelsbrunnerGS1986_opt_pl_mon_subdiv},
achieved $O(\log{n})$ query time as well.
For other methods and variants the reader is referred to
a comprehensive overview given in~\cite{Snoeyink2004_handbook_disc_comp_geom_pl}.



An asymptotically optimal solution is the randomized incremental
construction (RIC), which was introduced by
Mulmuley~\cite{Mulmuley1990_fast_planar_part_alg}
and Seidel~\cite{Seidel1991_simp_fast_inc_rand_td}.
In the static setting, it achieves $O(n\log{n})$ preprocessing time,
$O(\log{n})$ query time and $O(n)$ space, all in expectancy.
As pointed out in~\cite{CG-alg-app}, 
the latter two can even be worst-case guaranteed.
It is also claimed there that one can achieve these worst-case bounds 
in an expected preprocessing time of
$O(n\log^2{n})$, but no concrete proof is given.
The approach is able to handle dynamic scenes;
that is, it is possible to add or delete edges later on.
This method is discussed in more detail in Section~\ref{sec:RIC}.


\ignore{ 
A variant by Arya et al.~\cite{AryaMM2001_simp_entrpy_alg}
of the RIC adds weights and thus gives expected query timer
satisfying entropy bounds.
They also state that entropy preserving cuttings can be
used to give a method whose query time approaches the optimal entropy bound,
at the cost of increased space and programming complexity~\cite{AryaMM2001_entrpy_cutt}.
These methods guarantee a logarithmic query time, however maintaining the search structures
require a large amount of memory and a major increase in the preprocessing time.
Also the overall complexity of implementing these solutions is rather high.
Thus, we  not consider these solutions in our implementation.
}

{\bf Contribution:}
We present here a major revision of the trapezoidal-map random incremental
construction algorithm for \PPL in \CGAL. 
As the previous implementation, it provides a linear size data structure
for non-linear subdivisions that can handle static as well as dynamic scenes.
The new version is now able to guarantee $O(\log n)$ query time and $O(n)$ space.
Following recent changes in the ``2D Arrangements" package~\cite{bfhmw-scmtd-07},
the implementation now also supports unbounded subdivisions as
well as ones that are embedded on two-dimensional parametric surfaces.
After a review of the RIC in Section~\ref{sec:RIC}, we discuss,
in Section~\ref{sec:depth_vs_longest_query}, the difference
between the length~\lqpl of the longest search path and the depth~\depth of the DAG.
We prove that the worst-case ratio of~\depth and~\lqpl is~$\Theta(n/\log n)$.
Moreover, we describe two algorithms for the preprocessing stage
that achieve guaranteed $O(n)$ size and $O(\log n)$ query time.
Both are based on a verification of \lqpl after
the DAG has been constructed:
An implemented one that runs in expected $O(n \log^2 n)$ time,
and a more efficient one that runs in expected
$O(n\log n)$ time. The latter is a very recent addition
that was not included in the reviewed submission.
However, the solution that is integrated into CGAL is based on a
verification of~\depth.
Based on our experimental results, we conjecture that it
also achieves expected $O(n \log n)$ preprocessing time.
%
Section~\ref{sec:nn} demonstrates the performance of the
new implementation by comparing
our point location in a Voronoi Diagram with the nearest neighbor
implementation of the FDH and others.
Section~\ref{sec:impl_details} presents more details on the new
implementation. 
To the best of our knowledge, this is the only available
implementation for guaranteed logarithmic query time point location
in general two-dimensional subdivisions.  
\secspace\section{Review of the RIC for Point Location}\secspace
\label{sec:RIC}

We review here the random incremental construction (RIC) of an efficient \PL structure,
as introduced by~\cite{Mulmuley1990_fast_planar_part_alg,Seidel1991_simp_fast_inc_rand_td}
and described in~\cite{CG-alg-app,CG-intro-rand-alg}.
For ease of reading we discuss the algorithm in case the input is in general position.
Given an arrangement of $n$ pairwise interior disjoint $x$-monotone curves,
a random permutation of the curves is inserted incrementally,
constructing the Trapezoidal Map,
which is obtained by extending vertical walls
from each endpoint upward and downward
until an input curve is reached
or the wall extends to infinity.
During the incremental construction,
an auxiliary search structure, a directed acyclic graph (DAG), is maintained.
It has one root and many leaves, one for every trapezoid in the trapezoidal map.
Every internal node is a binary decision node, representing either
an endpoint $p$, deciding whether a query lies
to the left or to the right of the vertical line through $p$,
or a curve, deciding if a query is above or below it.
When we reach a curve-node, we are guaranteed that
the query point lies in the $x$-range of the curve.
The trapezoids in the leaves are interconnected, such that each trapezoid knows
its (at most) four neighboring trapezoids, two to the left and two to the right.
In particular, there are no common $x$-coordinates for two distinct
endpoints\footnote{In the general case all endpoints are lexicographically compared;
first by the $x$-coordinate and then by the $y$-coordinate.
This implies that two covertical points produce
a virtual trapezoid, which has a zero width.}.

When a new $x$-monotone curve is inserted,
the trapezoid containing its left endpoint is located
by a search from root to leaf.
Then, using the connectivity mechanism described above, the trapezoids intersected by the curve are
gradually revealed and updated.
Merging new trapezoids, if needed,
takes time that is linear in the number of intersected trapezoids.
The merge makes the data structure become a DAG (as illustrated in Figure~\ref{fig:basic_RIC_example})
with expected $O(n)$ size,
instead of an $\Omega(n\log{n})$ size binary tree~\cite{SeidelA2000_wc_query_comp}.
For an unlucky insertion order the size of the resulting data structure
may be quadratic, and the longest search path may be linear.
However, due to the randomization one can expect $O(n)$ space,
$O(\log{n})$ query time, and $O(n\log{n})$ preprocessing time.

\ignore{
As shown in~\cite{CG-alg-app}, this can even
be guaranteed by the following modification. While the structure is built,
one observes its size and the length of the longest search path~\lqpl.
As soon as one of the values becomes too large the
structure is rebuilt using a different random insertion order.

It is easy to show that only a small constant number of
rebuilds is expected~\cite{CG-alg-app}.
All proofs,~\cite{CG-alg-app,CG-intro-rand-alg},
are with respect to~\lqpl.
Thus, in order to retain the expected construction time of $O(n\log n)$
this value must be efficiently accessible. At present, it is not clear
how this can be achieved.
It is trivial to observe the
depth~\depth of the DAG, which is only an upper bound on~\lqpl.
Indeed, the original implementation of \CGAL's
RIC point location keeps track of~\depth to decide when to rebuild.
A detailed discussion of this problem can be found
below in Section~\ref{sec:depth_vs_longest_query}.
}

\begin{figure}[t]
  \vspace{-20pt}
  \begin{center}
    \begin{tabular}{cc}
      \includegraphics[width=0.48\textwidth]{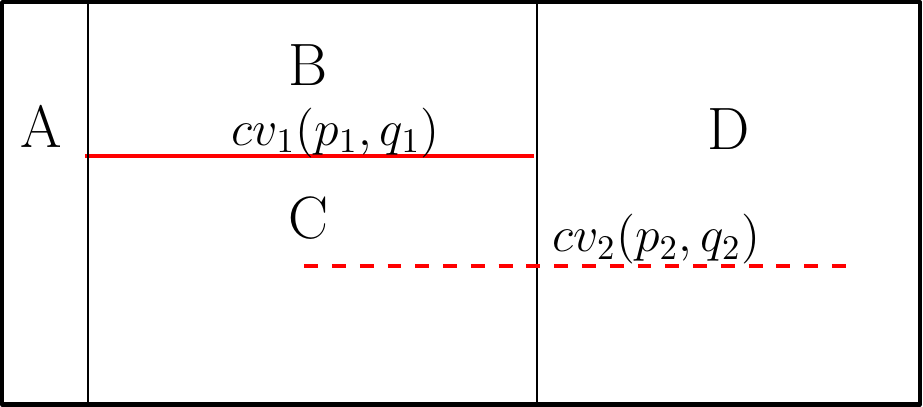} &
      \includegraphics[width=0.48\textwidth]{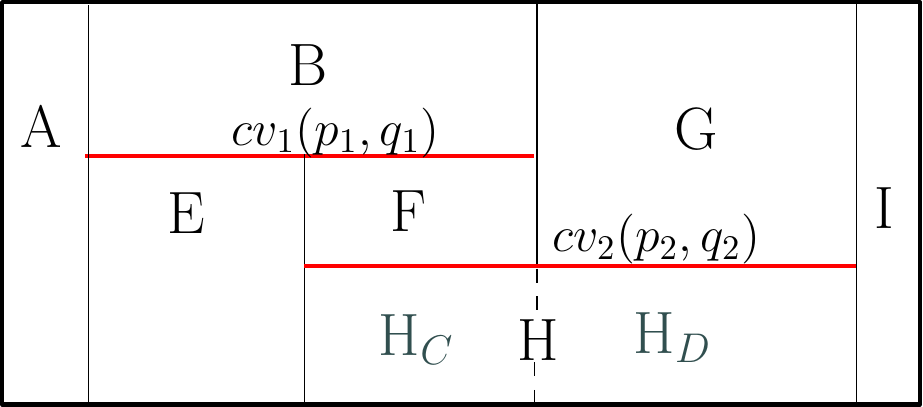} \\
      \includegraphics[width=0.48\textwidth]{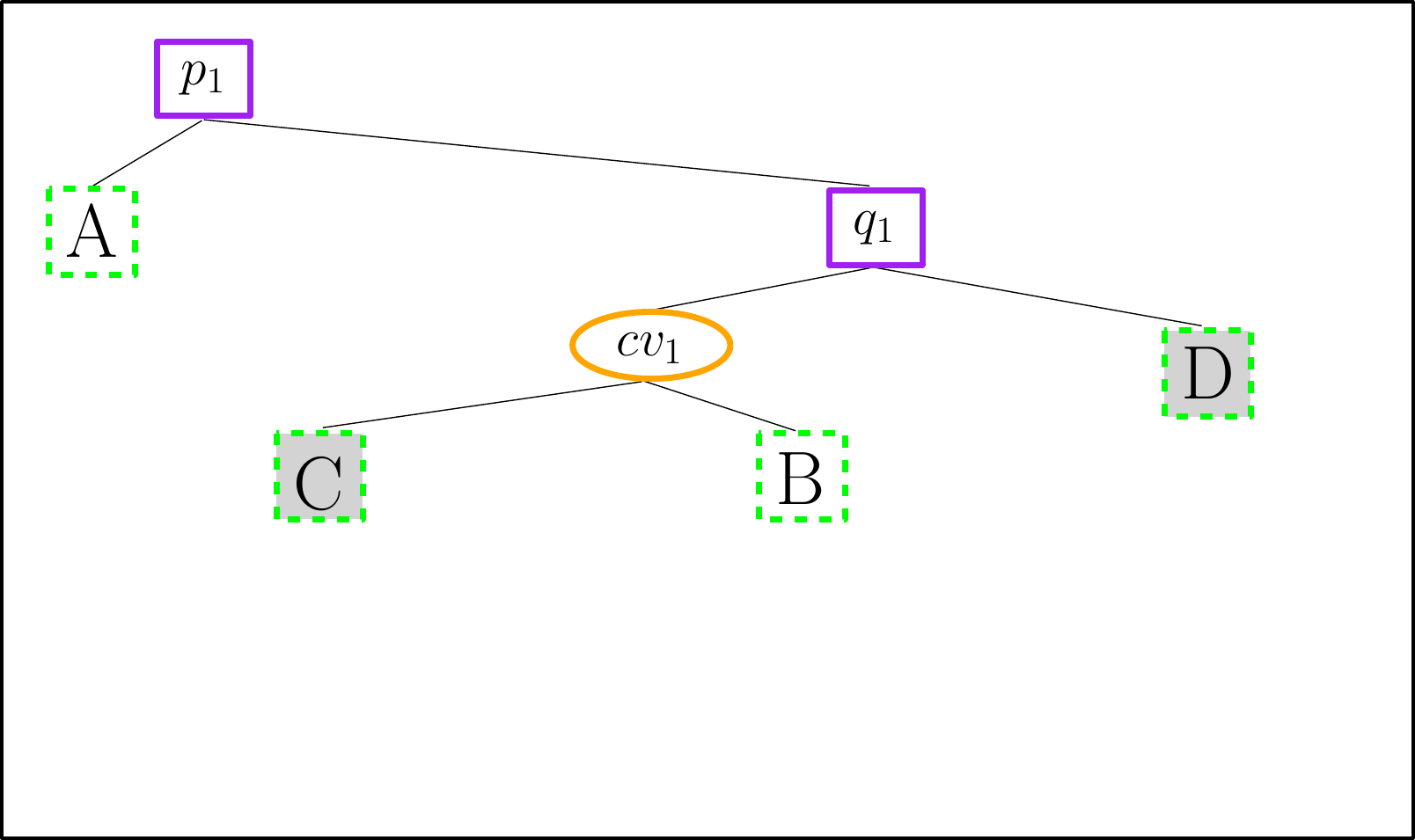} &
      \includegraphics[width=0.48\textwidth]{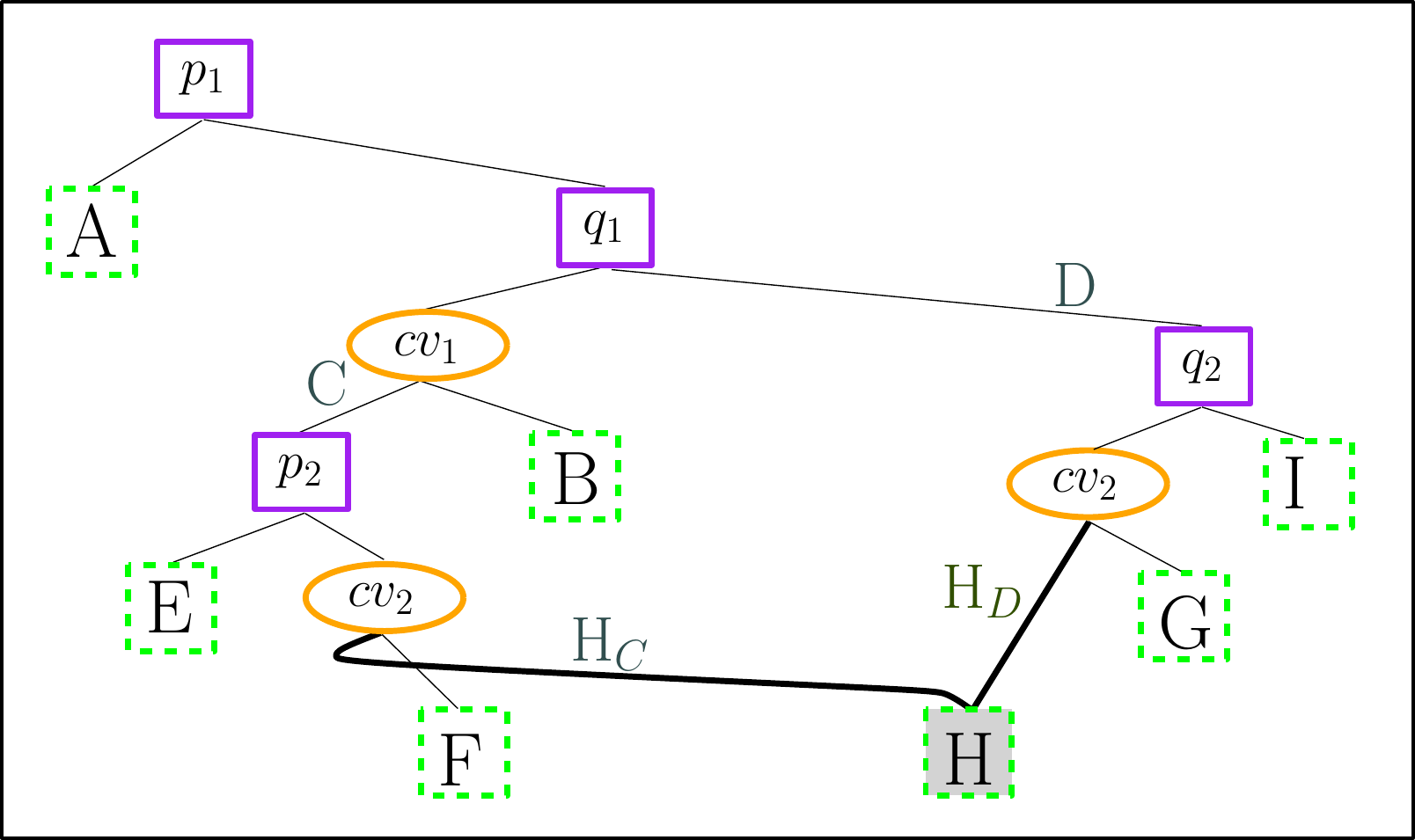} \\
      (a)&
      (b)
    \end{tabular}
  \end{center}
  \vspace{-10pt}
  \caption{
    Trapezoidal decomposition and the constructed DAG for two segments
    $cv_1$ and $cv_2$: 
    (a) before and (b) after the insertion of $cv_2$.
    The insertion of $cv_2$ splits the trapezoids $C,D$ into $E,F,H_C$ and $G,I,H_D$, respectively.
    $H_C$ and $H_D$ are merged into $H$, as they share the same top (and bottom) curves.
  }
\ignore{
  \caption{The trapezoidal decomposition and the constructed DAG
    for the segments $cv_1$ and $cv_2$ ($cv_i$ has endpoints $p_i$ and $q_i$).
    (a) After the insertion of the first curve $cv_1$.
    (b) After the insertion of the second curve $cv_2$.
    First, the trapezoid $C$ containing $p_2$, the left endpoint of $cv_2$, is found.
    Then its right neighbor $D$, intersected by $cv_2$, is reached.
    Trapezoids $A,B$ remain unchanged.
    Trapezoids $C,D$ are split into $E,F,H_C$ and $G,I,H_D$, respectively.
    $H_C$ and $H_D$ share the same top and bottom curves,
  therefore are merged into one trapezoid $H$.}
}
\vspace{-16pt}
\label{fig:basic_RIC_example}
\end{figure} 
\secspace\section{On the Difference between Paths and Search Paths}\secspace

\label{sec:depth_vs_longest_query}

As shown in~\cite{CG-alg-app}, one can build a data structure,
which guarantees $O(\log{n})$ query time and $O(n)$ size,
by monitoring the size and the length of the longest search path~\lqpl
during the construction.
\new The idea is that as soon as one of the values becomes too
large, the structure is rebuilt using a different random insertion order.
It is shown that only a small constant number of
rebuilds is expected.
However, in order to retain the expected construction time of $O(n\log n)$,
both values must be efficiently accessible.
While this is trivial for the size, it is not
clear how to achieve this for~\lqpl.
Hence, we resort to the depth~\depth of the DAG, which
is an upper bound on~\lqpl as the set of all possible
search paths is a subset of all paths in the DAG.
Thus, the resulting data
structure still guarantees a logarithmic query time.

The depth~\depth can be made accessible in constant time
by storing the depth of each leaf in the leaf itself, and maintaining
the maximum depth in a separate variable.
The cost of maintaining the depth can be charged to new nodes,
since existing nodes never change their depth value.
This is not possible for~\lqpl while retaining linear space,
since each leaf would have to store a non-constant number of values, i.e.,
one for each valid search path that reaches it.
In fact the memory consumption would be equivalent to the data structure that
one would obtain without merging trapezoids, namely the trapezoidal search tree,
which for certain scenarios requires $\Omega(n\log{n})$ memory as shown
in~\cite{SeidelA2000_wc_query_comp}.
In particular, it is necessary to merge as (also in practice) the sizes of the
resulting search tree and the resulting DAG considerably
differ.

In Section~\ref{ssub:depth_path_ratio} we show that for a given DAG
its depth~\depth can be linear while~\lqpl is still logarithmic, that is, such a
DAG would trigger an unnecessary rebuild. It is thus questionable whether
one can still expect a constant number of rebuilds when relying on~\depth.
\new Our experiments in Subsection~\ref{ssub:dept_path_ratio_exp} show
that in practice the two values hardly differ,
which indicates that it is sufficient to rely on~\depth.
%
%
However, a theoretical proof to consolidate this is still missing.
Subsection~\ref{ssec:static_sub_sol} provides efficient
preprocessing solutions for the static scenario (where all segments
are given in advance). As such, we see it as a concretization of, and
an improvement over, the claim mentioned in~\cite{CG-alg-app}.


\subsecspacea\subsection{Worst Case Ratio of Depth and Longest Search Path}\subsecspaceb
\label{ssub:depth_path_ratio}

\begin{wrapfigure}{r}{0.547\textwidth}
  \vspace{-24pt}
  \includegraphics[width=0.546\textwidth]{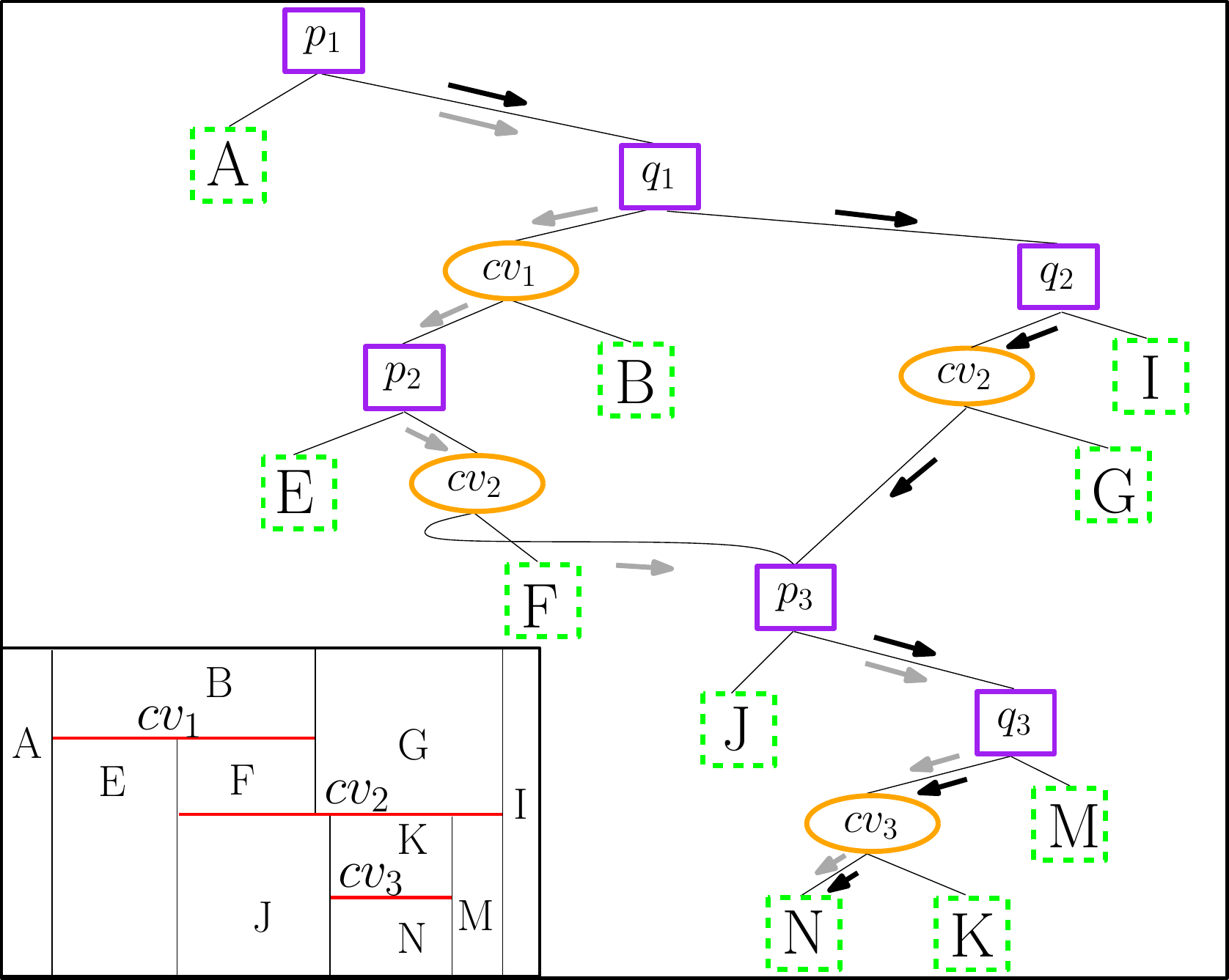}
  \vspace{-30pt}
  \label{fig:dag_paths}
\end{wrapfigure}
The figure to the right shows the DAG of
Figure~\ref{fig:basic_RIC_example} after inserting a
third segment. There are two paths that reach the trapezoid
$N$ (black and gray arrows).
However, the gray path is not a valid search path, since
all points in $N$ are to the right of~$q_1$; that is, such a search
would never visit the left child of $q_1$. It does, however,
determine the depth of~$N$, since it is the longer path of the two.
In the sequel we use this observation to construct an example
that shows that the ratio between~\depth and~\lqpl can be
as large as $\Omega(n/\log n)$.
Moreover, we will show that this bound is tight.

\begin{wrapfigure}{r}{0.4\textwidth}
  \vspace{-22pt}
    \includegraphics[width=0.4\textwidth]{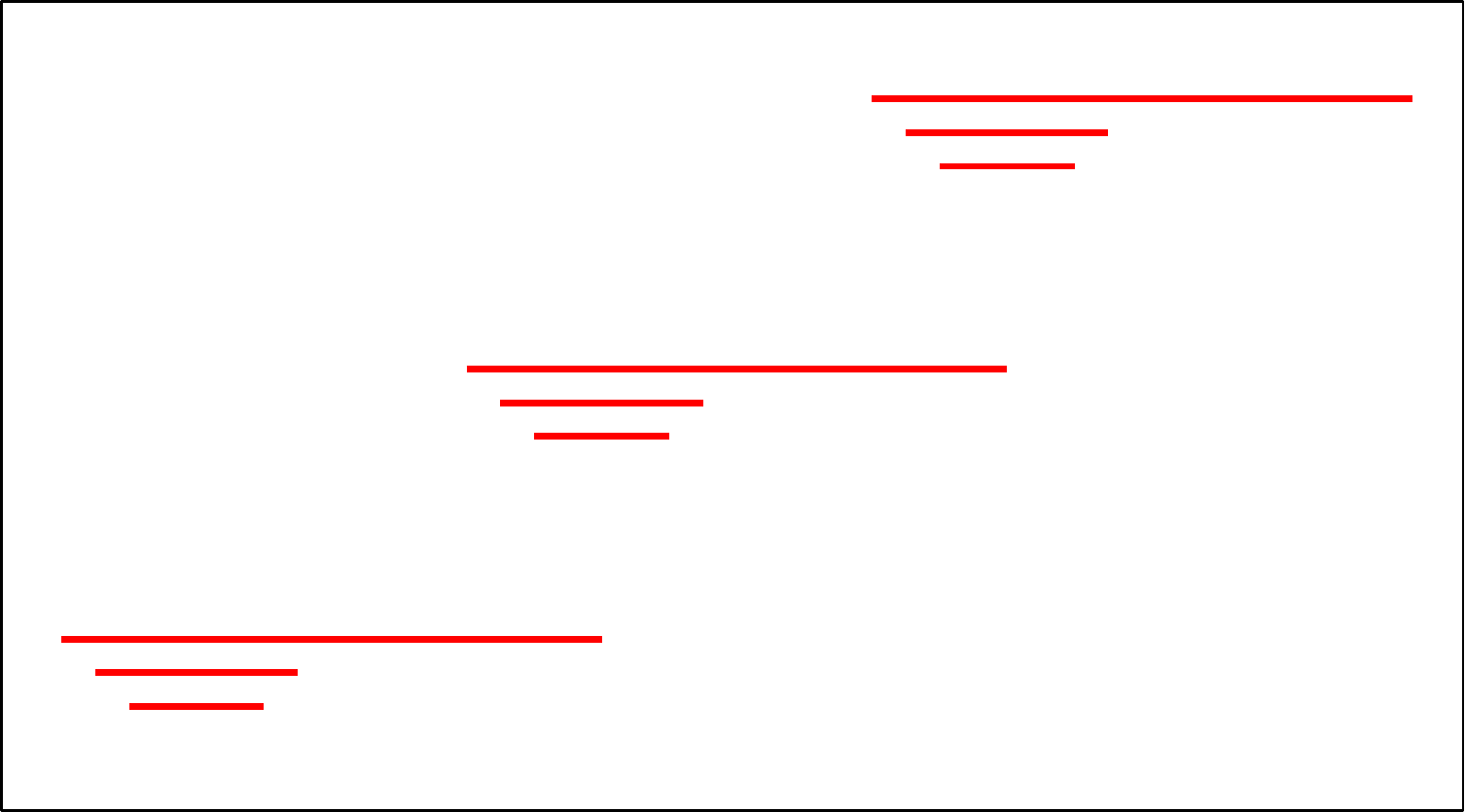}
  \vspace{-32pt}
  \label{fig:sqrt_query_linear_depth}
\end{wrapfigure}
We start by constructing a simple-to-demonstrate lower bound that achieves $\Omega(\sqrt{n})$
ratio between \depth and \lqpl.
Assuming that $n = k^2 \in \N$, the construction consists of $k$ blocks,
each containing $k$ horizontal segments.
The blocks are arranged as depicted in the figure to the right.
Segments are inserted from top to bottom.
A block starts with a large segment at the top, which we call the {\it cover segment},
while the other segments successively shrink in size.
Now the next block is placed to the left and below the previous block.
Only the cover segment of this block extends below the previous block,
which causes a merge as illustrated in Figure~\ref{fig:sqrt_query_linear_depth_before_after}.
All $k=\sqrt{n}$ blocks are placed in this fashion.
This construction ensures that each newly inserted segment intersects the
trapezoid with the largest depth, which increases~\depth.
The largest depth of $\Omega(n)$ is finally achieved in the
trapezoid below the lowest segment.
However, the actual search path into this trapezoid
has only $O(\sqrt{n})$ length, since for each previous block it only passes
through one node in order to skip it and $O(\sqrt{n})$ in the last block.



\begin{figure}[t]
    \centering
    \setlength{\tabcolsep}{3pt}
    \begin{tabular}{c c}
        \includegraphics[height=5.3cm]{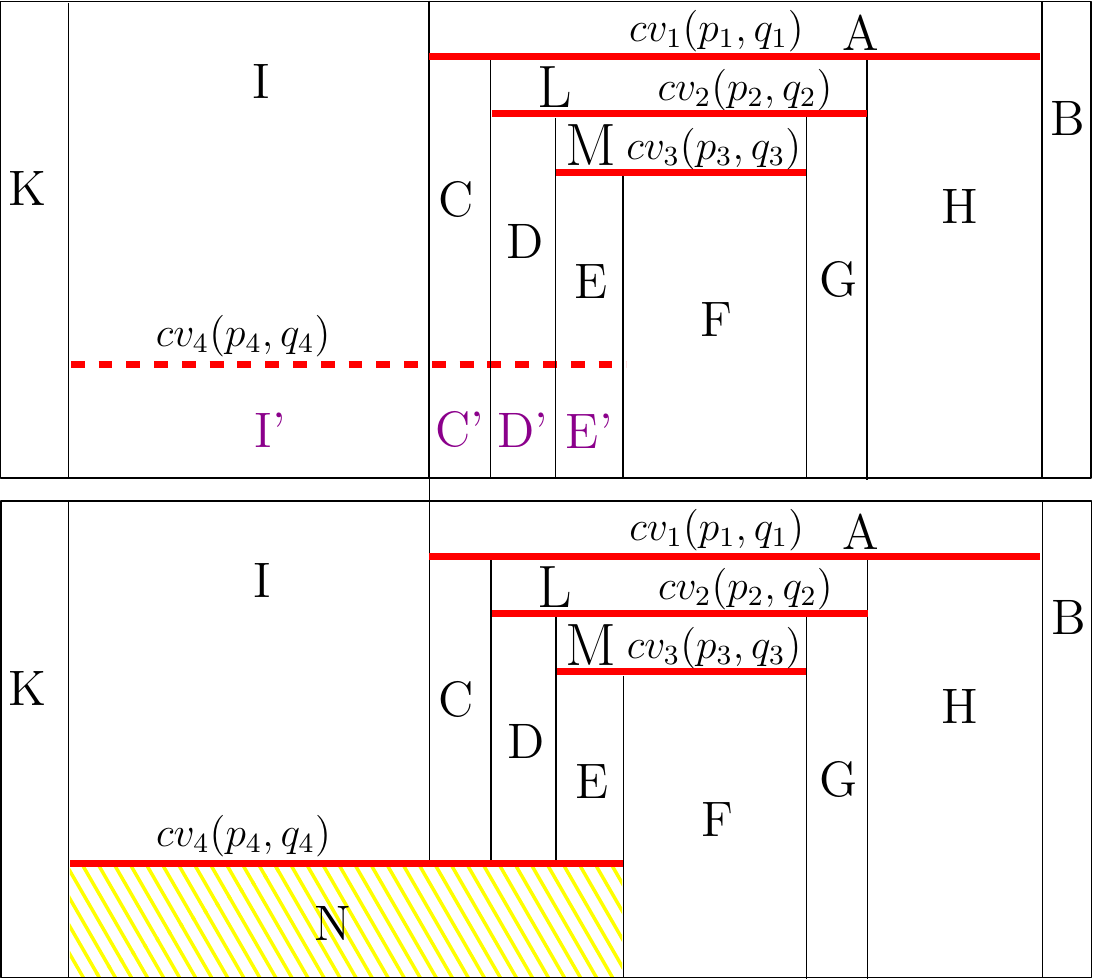} &
        \includegraphics[height=5.3cm]{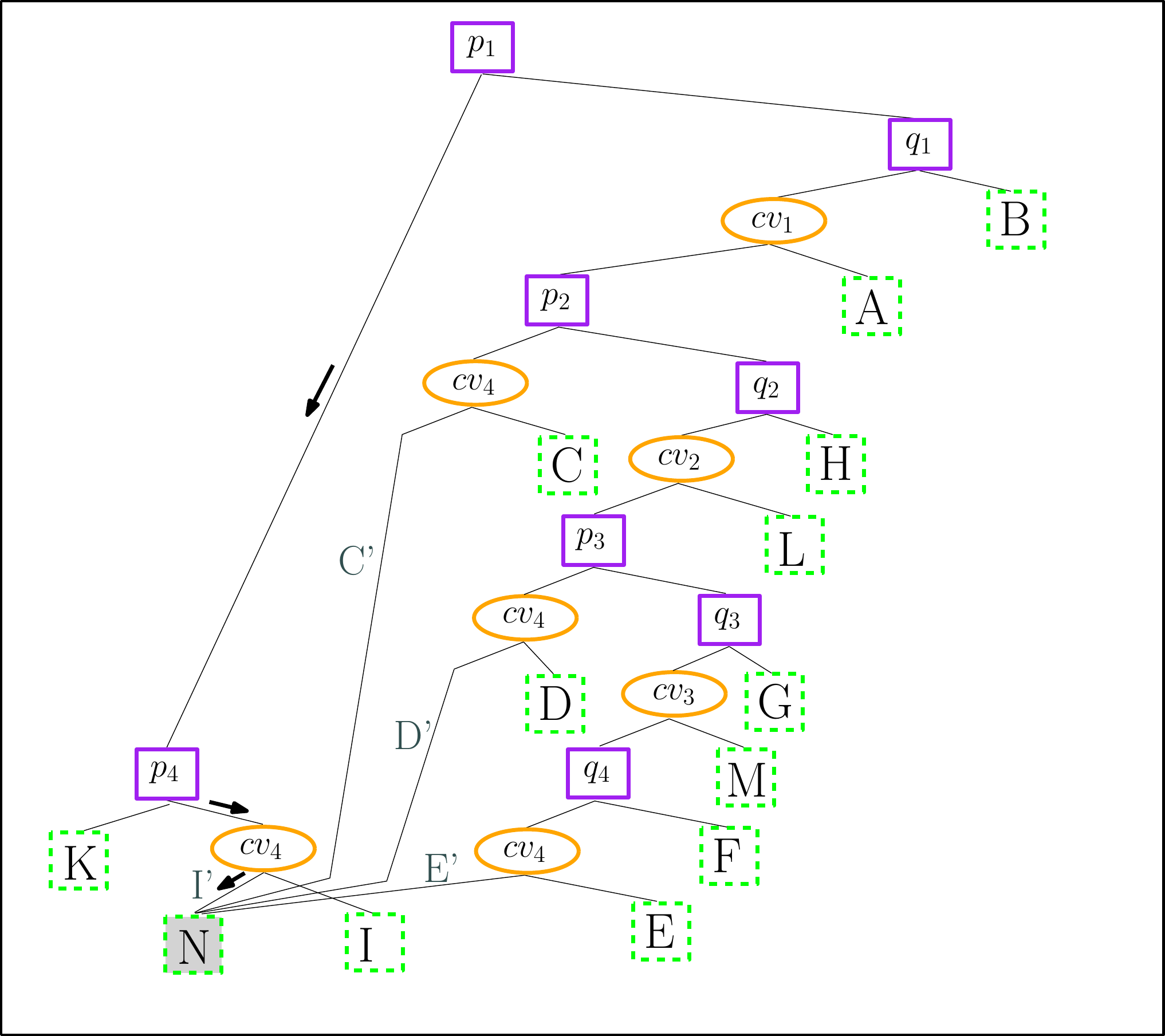}\\
        (a) &
        (b)
    \end{tabular}
    \vspace{-5pt}
    \caption{(a) The trapezoidal-map after inserting $cv_4$.
    The map is displayed before and after the
    merge of $I'$, $C'$, $D'$, and $E'$ into $N$,
    in the top and bottom illustrations, respectively.
    A query path to the region of $I'$ in $N$ will take 3 steps,
    while the depth of $N$ in this example is 11.}
    \label{fig:sqrt_query_linear_depth_before_after}
    \vspace{-10pt}
\end{figure}

\begin{wrapfigure}{r}{0.4\textwidth}
  \vspace{-32pt}
  \begin{center}
    \includegraphics[width=0.4\textwidth]{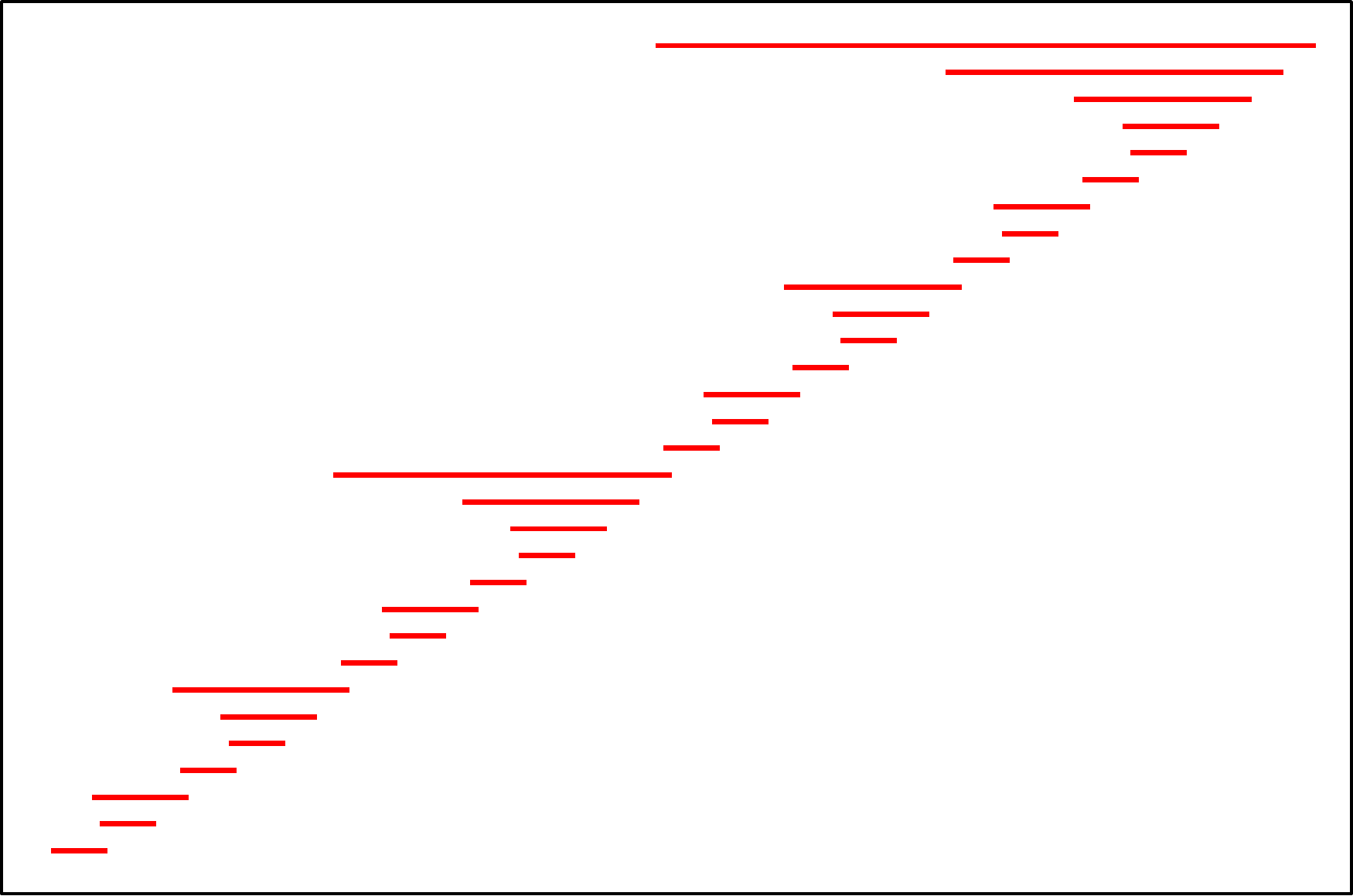}
  \end{center}
  \vspace{-40pt}
  \label{fig:depth_vs_log_qlen}
\end{wrapfigure}
The following construction, which uses a recursive scheme,
establishes the lower bound $\Omega(n/\log{n})$ for~\depth/\lqpl.
Blocks are constructed and arranged in a similar fashion as in the
previous construction. However, this time we have $\log n$ blocks,
where block $i$ contains $n/2^i$ segments. Within each block
we then apply the same scheme recursively as depicted in the
figure to the right.
Again segments are inserted top to bottom such that
the depth of $\Omega(n)$ is achieved in the trapezoid below the
lowest segment.
\second{
The fact that the lengths of
all search paths are logarithmic can be easily proven
by applying a simple induction argument.
}{%
The fact that the lengths of all search paths are logarithmic
can be proven by the following argument.
By induction we assume that the longest path within a block of
size $n/2^i$ is some constant times $(\log_2{n} - i)$.
Obviously this is true for a block containing only one segment.
Now, in order to reach block $i$ with $n/2^i$ segments,
we require $i-1$ comparisons to skip the $i-1st$ preceding blocks.
Thus in total the search path is of logarithmic length.
}

\begin{theorem}
The $\Omega(n/\log{n})$ worst-case lower bound on~\depth/\lqpl is tight.
\end{theorem}
\begin{proof}
Obviously,~\depth of $O(n)$ is the maximal achievable depth, since
by construction each segment can only appear once
along {\em any} path in the DAG.
It remains to show that for any scenario with $n$ segments
there is no DAG for which~\lqpl
is smaller than $\Omega(\log{n})$.
Since there are $n$ segments, there are at least $n$ different
trapezoids having these segments as their top boundary.
Let $T$ be a decision tree of the optimal search structure.
Each path in the decision tree corresponds to a valid search path
in the DAG and vice versa.
The depth of $T$ must be larger than $\log_2{n}$, since it is only
a binary tree.
We conclude that the worst case ratio of~\depth and~\lqpl is
$\Theta(n/\log{n})$.
\qed
\end{proof}

\subsecspacea\subsection{Efficient Solutions for Static Subdivisions}\subsecspaceb
\label{ssec:static_sub_sol}
We first describe an algorithm for static scenes that runs in
expected $O(n\log^2{n})$ time, constructing a DAG of linear size
in which \lqpl is $O(\log{n})$.
%
The result is based on the following lemma.


\begin{lemma}
\label{lemma:tree_size}
Let $S$ be a planar subdivision induced by $n$ pairwise interior disjoint $x$-monotone curves.
The expected size of the trapezoidal search tree $\cal T$,
which is constructed as the RIC above but without merges,
is $O(n\log{n})$.
\end{lemma}


\begin{proof}
Since $\cal T$ is binary tree, it is sufficient to bound the expected number
of leaves in $\cal T$, namely, the number of trapezoids (without merges),
which is bounded by twice the number of vertical edges + 1.
First, focus on a vertical wall $W$ induced by one endpoint of the
$i$th inserted curve.
$W$ is intersected by $n$ curves, in the worst-case.
The $i-1$ already inserted curves partition $W$ into $i$ intervals.
However, we are only interested in the interval $I$ containing the
endpoint of the $i$th curve, as it will
appear in the final structure.
Curves inserted after the $i$th curve may split $I$.
The expected number of intersections in $I$
(including the endpoint of the $i$th curve) is $O((n-i)/i)$.
Summing up over all vertical walls gives a total of $O(n\log{n})$ expected intersections.
Thus, the expected number of vertical edges is $O(n\log{n})$ as well,
and , clearly, this is also the expected size of the tree.
\qed
\end{proof}

The following algorithm \emph{compute\_max\_search\_path\_length} computes~\lqpl in
expected $O(n\log^2{n})$ time.
Starting at the root it descends towards the leaves in a recursive fashion.
Taking the history of the current path into account, each recursion call
maintains the interval of $x$ values that are still possible.
Thus, if an $x$-coordinate of a point node is not contained in the interval
the recursion does not need to split.
This means that the algorithm essentially mimics $\cal T$
\new (as it would have been constructed),
since the recursion only follows possible search paths.
By Lemma~\ref{lemma:tree_size} the expected number of leaves of $\cal T$,
and thus of search paths, is $O(n\log{n})$.
Since the expected length of a query is $O(\log{n})$ this algorithm
takes expected $O(n\log^2{n})$ time.

\mike{The amount of additional memory is  $O(\log{n})$ since the recursion
corresponds to a DFS in the trapezoidal search tree, that is, at any time at
most $O(\log{n})$ functions are on the call stack.
}


\begin{definition} $f(n)$ denotes the time it
takes to verify that, in a linear size DAG constructed over a planar
subdivision of $n$ $x$-monotone curves, \lqpl is bounded by $c \log n$
for a constant $c$.
\end{definition}

\begin{theorem}
\label{thrm:gnrl_preproc_time}
Let $S$ be a planar subdivision with $n$ $x$-monotone curves.
A \PL data structure for $S$, which has
$O(n)$ size and $O(\log{n})$ query time in the worst case,
can be built in $O(n\log{n} + f(n))$ expected time, where
$f(n)$ is as defined above.
\end{theorem}

\begin{proof}
The construction of a DAG with some random insertion order takes expected
$O(n\log n)$ time. The linear size can be verified trivially on the fly.
After the construction the algorithm \emph{compute\_max\_search\_path\_length}
is used to verify that~\lqpl is logarithmic.
The verification of the size and~\lqpl may trigger rebuilds with a new
random insertion order. However, according to~\cite{CG-alg-app}, one can expect
only a constant number of rebuilds.
Thus, the overall expected runtime remains expected $O(n\log{n} + f(n))$.
\qed
\end{proof}

The verification process described above takes expected $O(n \log^2 n)$
time. However, one can do better as we briefly sketch next. Let $T$
be the collection of {\it all} the trapezoids created during the
construction of the DAG, including intermediate trapezoids that are
later killed by the insertion of later segments. Let ${\cal A}(T)$
denote the arrangement of the trapezoids. The {\it depth} of a point
$p$ in the arrangement is defined as the number of trapezoids in $T$
that cover $p$.
The key to the improved algorithm is the following observation by
Har-Peled.
\\

\noindent
{\bf Observation 1.} {\em The length of a path in the DAG for a query point
$p$ is at most three times the depth of $p$ in ${\cal A}(T)$.}\\

We remark that this depth is established in an
interior of a face of ${\cal A}(T)$ since we consider the boundaries
of the trapezoids as open.
This can be done since the longest path will always end in a leaf of the DAG,
which represents a trapezoid.
For any query point that falls on either a curve or an endpoint 
of the initial subdivision
the search path will end in an internal node of the DAG.
The search path for a query point $q$ on a vertical edge
of a trapezoid will be identical to a path
for a query point in a neighboring trapezoid.

It follows that we need to verify that the maximum depth of a point
in ${\cal A}(T)$ is some constant $c_1 \log n$.
Since the input curves in $S$ are interior
pairwise disjoint, according to the separation property stemming from \cite{GY-TSR-80},
one can define a total order on the curves.
This order allows us to apply a modified version%
\footnote{More details can be found in Appendix~\ref{sec:comp_arrangement_depth}.}
of an algorithm by Alt and Scharf~\cite{as-cdaaa-13}, 
which originally detects the
maximum depth in an arrangement of $n$ 
axis-parallel rectangles in $O(n \log n)$ time.
%
%
Notice that we only apply this verification algorithm on DAGs of
linear size.
Putting everything together we obtain:
\begin{theorem}
\label{thrm:new}
Let $S$
be a planar subdivision with $n$ $x$-monotone curves. A \PL data
structure for $S$, which has $O(n)$ size and $O(\log{n})$ query time
in the worst case, can be built in $O(n\log{n})$ expected time.
\end{theorem}

\ignore{
\paragraph{{\bf Remark:}} The recursion in
\emph{compute\_max\_search\_path\_length}
essentially mimics a DFS traversal of the trapezoidal search tree,
which has only expected $O(n\log{n})$ edges.
Thus, one could assume that the recursion also
takes expected $O(n\log{n})$ time.
However, paths traversed by the algorithm usually have
additional nodes, namely those that represent endpoints,
which are not contained in the $x$-range of the interval.
\new 
}

\ignore{
One may simply observe the length of the path that was taken in order
to find the left endpoint of a newly inserted segment. Together with the fact that one also
ensures that the size of the DAG is linear it would be ensured that each build takes
$O(n\log n)$ time. However, the resulting DAG may even contain a search path of linear
size. Simply, take the recursive construction for the lower bound on~\depth/\lqpl
and extend all segments to the right. Clearly locating the left endpoints of newly inserted
segments takes $O(\log n)$ time. However, since they all extend to the right a search
path in this region is linear. But at least this triggers a solution for a static
subdivision.

The above algorithm ensures that the construction of a DAG takes $O(n\log n)$ time.
As this DAG may still contain paths of linear length we close the construction
by calling \emph{Ensure\_path\_length}. If the predicate returned \emph{false},
we rebuild the DAG.

\begin{lemma}
TODO INSERT LEMMA HERE
\end{lemma}

As we see it, this lemma adds the additional layer
needed for the completeness of the algorithm,
as described in~\cite{CG-alg-app}, chapter X.
}

\ignore{
The depth~\depth can be maintained such that it is  accessible in constant time.
Each leaf stores its depth, making it possible to compute the depth of new leafs.
The maximum depth~\depth is maintained in a separate variable and updated if necessary.
However, it is not possible to do the same for~\lqpl.
The reason is that each leaf would have to store at least one value for each valid
search path that reaches the leaf.
The number of such paths in not bounded by a constant.
In fact the memory consumption would be equivalent to the data structure that one obtains
without merging trapezoids, namely the trapezoidal search tree, which
requires $\Omega(n\log{n})$ memory as shown in~\cite{SeidelA2000_wc_query_comp}.
In particular, it is necessary to merge as the sizes of the resulting search tree and the
resulting DAG considerably differ also in practice,
as demonstrated in Table~\ref{tbl:tree_dag_size}.
\begin{table}[t]
  \vspace{-20pt}
  \caption{
    Number of tree nodes vs. number of DAG nodes for the 
    same input with the same insertion order.}
  \begin{center}
    \begin{tabular}{| c | c | c | c |}
    \hline
    \# Arrangement Edges & \# Tree nodes & \# DAG nodes & ratio\\
    \hline
    138 & 1263 & 681 & 1.85\\
    \hline
    285 & 3167 & 1492 & 2.12\\
    \hline
    1483 & 23511 & 8019 & 2.93\\
    \hline
    2977 & 51551 & 16330 & 3.15\\
    \hline
    14975 & 350629 & 84576 & 4.14\\
    \hline
    29973 & 759075 & 169355 & 4.48\\
    \hline    \end{tabular}
  \end{center}
  \label{tbl:tree_dag_size}
  \vspace{-20pt}
\end{table} 

On the other hand it is actually not necessary to have~\lqpl accessible in constant time.
First, the insertion of an edge takes $O(\log n)$ time since this requires to locate the
left most end point of the edge in the trapezoidal map.
Second, the insertion (or deletion) of one edge can only increase~\lqpl by a
constant, which implies that we only have to check~\lqpl
after $O(\log n)$ insertions (deletions).
Giving us $O(\log^2{n})$ time to compute~\lqpl.
}

\subsecspacea\subsection{Experimental Results}\subsecspaceb
\label{ssub:dept_path_ratio_exp}

\begin{figure}[t]
    \includegraphics[width=0.48\textwidth]{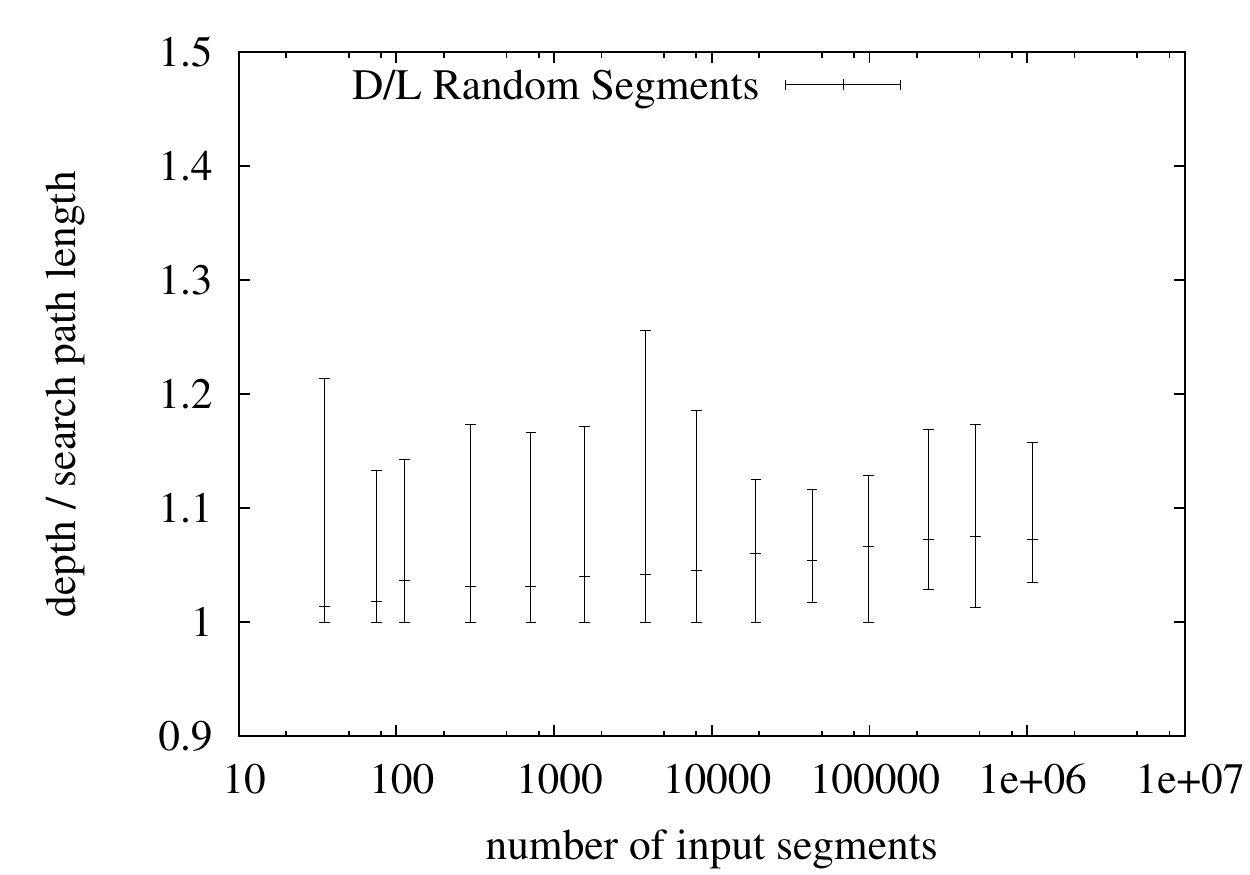}
    \includegraphics[width=0.48\textwidth]{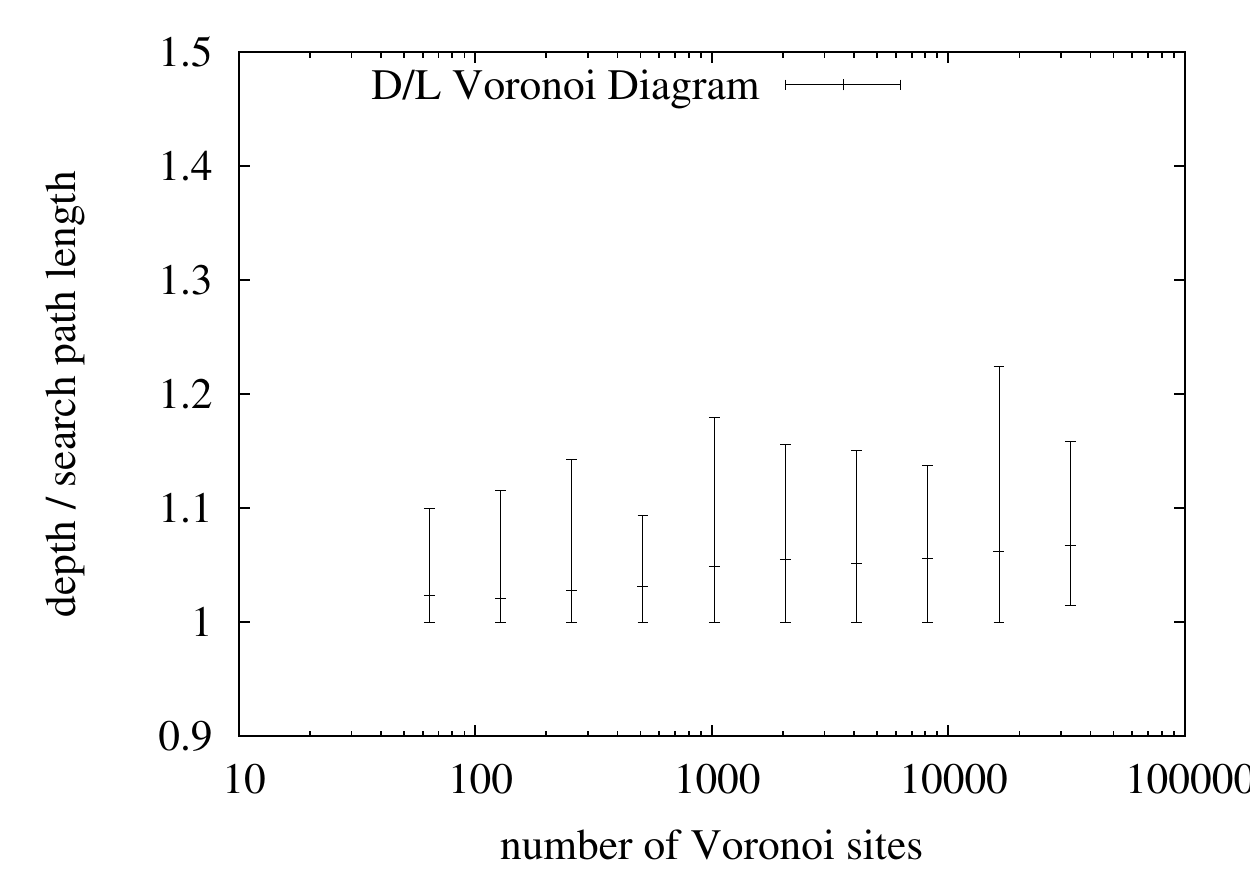}
    \caption{
      Ratio of~\depth and~\lqpl in two scenarios:
      random segments (left), Voronoi diagram of random points (right).
      Plots show average value with error bars.
    }
    \label{fig:depth_vs_length_plot}
\end{figure}

Since \depth is an upper bound on \lqpl and since \depth is accessible
in constant time our implementation explores an alternative that
monitors~\depth instead of~\lqpl.
Though this may cause some additional rebuilds,
the experiments in this section give strong evidence
that one can still expect $O(n\log n)$ preprocessing time.
We compared~\depth and~\lqpl in two different scenarios:
random non-intersecting line segments and
Voronoi diagram for random sites.%
%
\footnote{Appendix~\ref{sec:dl_ratio_results_appndx} contains additional experimental
results that include also the scenarios constructed in Section~\ref{ssub:depth_path_ratio}.}
%
Each scenario was tested with an increasing number of subdivision edges,
with several runs for each input.
Figure~\ref{fig:depth_vs_length_plot} displays
the average~\depth/\lqpl ratio, and also the minimal and maximal ones.
Obviously, the average ratio is close to~1 and never exceeded a value of~$1.3$.

These experimental results indicate that replacing the test for the length of the
longest path~\lqpl by the depth~\depth of the DAG in the randomized
incremental construction essentially does not harm the runtime.
However, the following conjecture remains to be proven.

\begin{conjecture}
\label{con:depth}
There exists a constant $c>0$ such that the runtime of the randomized
incremental algorithm, modified such that it rebuilds in case the depth~\depth
of the DAG  becomes larger than $c\log n$, is expected $O(n\log n)$,
i.e., the number of expected rebuilds is still constant.
\end{conjecture}

\ignore{
This indicates that in practice one can safely replace~\lqpl by~\depth without
harming the expected runtime of the algorithm. In particular, this can easily
be compensated by a slight increase of the threshold.

\begin{proposition}
Replacing the check for the length of the longest path~\lqpl by the
depth~\depth of the DAG in the the randomize incremental construction
does not harm the expected runtime of $O(n\log n)$, i.e.,
one can still expect a constant number of rebuilds.
\end{proposition}
}

\ignore{
\mike{Ok, sorry, that does not work too. The constants make it too precise.
After all it's not a theorem.}
\begin{proposition}
Let $S$ be a set of $n$ non-intersecting curves, let $c$ be a small constant,
and let $\gamma$ be constant s.t. $\gamma < 0.25$.
Then the probability that the depth~\depth of the
DAG, constructed by the RIC for \PL algorithm for $S$,
is more than $c\log(n+1)$ is at most $\gamma$.
\end{proposition}

In other words, we conjecture that the expected
DAG depth~\depth is logarithmic as well.
If this conjecture is proven to be true,
one could replace~\lqpl with~\depth and use a
slightly higher threshold in the validation conditions
 used in the construction.
\michal{The three ratio tables(~\ref{tbl:rand_segs_ratio},
~\ref{tbl:rand_sites_ratio}, ~\ref{tbl:circle_sites_ratio})
should be removed and replaced by the plot in figure~\ref{fig:depth_vs_length_plot}.
This is only a sketch of a plot, the values are real but more values are needed.
Wrong axes names
}
}

\ignore{
\subsecspacea\subsection{Alternative Approaches}\subsecspaceb
\label{ssec:depth_alternatives}

\mike{needs restructure. There are three things to be presented: How to compute L; idea of
only checking insertion paths including counter example; solution for static scene which is
a combination of the previous two. }



The depth~\depth can be maintained such that it is  accessible in constant time.
Each leaf stores its depth, making it possible to compute the depth of new leafs.
The maximum depth~\depth is maintained in a separate variable and updated if necessary.
However, it is not possible to do the same for~\lqpl.
The reason is that each leaf would have to store at least one value for each valid
search path that reaches the leaf.
The number of such paths in not bounded by a constant.
In fact the memory consumption would be equivalent to the data structure that one obtains
without merging trapezoids, namely the trapezoidal search tree, which
requires $\Omega(n\log{n})$ memory as shown in~\cite{SeidelA2000_wc_query_comp}.
In particular, it is necessary to merge as the sizes of the resulting search tree and the
resulting DAG considerably differ also in practice,
as demonstrated in Table~\ref{tbl:tree_dag_size}.
\begin{table}[t]
  \vspace{-20pt}
  \caption{
    Number of tree nodes vs. number of DAG nodes for the 
    same input with the same insertion order.}
  \begin{center}
    \begin{tabular}{| c | c | c | c |}
    \hline
    \# Arrangement Edges & \# Tree nodes & \# DAG nodes & ratio\\
    \hline
    138 & 1263 & 681 & 1.85\\
    \hline
    285 & 3167 & 1492 & 2.12\\
    \hline
    1483 & 23511 & 8019 & 2.93\\
    \hline
    2977 & 51551 & 16330 & 3.15\\
    \hline
    14975 & 350629 & 84576 & 4.14\\
    \hline
    29973 & 759075 & 169355 & 4.48\\
    \hline    \end{tabular}
  \end{center}
  \label{tbl:tree_dag_size}
  \vspace{-20pt}
\end{table} 

On the other hand it is actually not necessary to have~\lqpl accessible in constant time.
First, the insertion of an edge takes $O(\log n)$ time since this requires to locate the
left most end point of the edge in the trapezoidal map.
Second, the insertion (or deletion) of one edge can only increase~\lqpl by a
constant, which implies that we only have to check~\lqpl
after $O(\log n)$ insertions (deletions).
Giving us $O(\log^2{n})$ time to compute~\lqpl.

However, the algorithm to compute~\lqpl that we propose here requires $O(n\log n)$ time.
Essentially it computes all possible search paths in the DAG.
It does so by discarding those paths that are geometrically impossible.
Starting at the root it descends towards the leaves in a recursive fashion, taking the
history of the current path into account it maintains the interval of $x$ values that
are still possible. This way the recursion only needs to
descend in one direction if the $x$ value of a vertex node is not contained in the
maintained interval as only one side is geometrically possible.
Thus, the algorithm essentially mimics the trapezoidal search tree and takes
$O(n\log n)$ time.
Obviously, it is not an option to call this algorithm after every insertion or even
only after every $O(\log n)$ insertions.

\ignore{
\subsubsection{Avoiding Merges}
This would turn the DAG into a tree and every path in this tree would
also be a possible search path. Thus, applying a similar scheme as described
for the depth above would make~\lqpl available 
in constant time. However, this is not an option since it is
known~\cite{SeidelA2000_wc_query_comp} that the size of the resulting tree
is $\Omega(n\log{n})$ and not linear. This increased memory consumption is
also relevant in practice as illustrated in Table~\ref{tbl:tree_dag_size}.
Thus, the merge steps are required.

We remark that it is also not an option to maintain the depth values of
all preceding nodes in a leaf node as the number of these values is not
constant. In particular, one would have to store the same information as
the tree and the size of the resulting data structure would be again in
$\Omega(n\log n)$.
}

\ignore{
If the data structure of the DAG is not modified the most efficient
algorithm that computes the~\lqpl essentially mimics the trapezoidal search tree.
It does so by computing all possible search paths in the DAG while discarding
those that are geometrically not possible. Starting at the root it descends towards the leaves
in a recursive fashion, taking the history of the current path into account it maintains the
interval of $x$ values that are still possible. In particular, the recursion only needs to
descend in one direction if the $x$ value of a vertex node is not contained in the
maintained interval.

Another algorithm, which as far as we know is the most efficient, would
go over the DAG and mimic the trapezoidal search tree.
It begins at the root, with minus infinity as min-boundary and plus
infinity as max-boundary,and recursively checks the two children, taking
the maximal query length between them and adds 1.
The recursion stops when a leaf is found, or if the min-boundary is not
lexicographically smaller than the max-boundary.
An above-below node on the way continues recursively with the two children,
with the same boundaries.
A left-right node however, continues recursively with updated boundaries.
Let $p$ be the underlying endpoint, the recursion continues with $p$ as the
max-boundary for the left child , and as the max-boundary for the right child.
In fact, the number of recursion paths ending in a certain trapezoid is the
exact number of slabs combining the trapezoid in the matching trapezoid search tree.
The complexity of this algorithm is correlated to the structure of the
traversed DAG and of the matching tree, thus is expected $O(n\log{n})$.
For a worst case structure the running time will obviously take longer.

One may only compute the value in the end of the construction in order to
make sure the DAG does not violate the requirements.
This would not effect the preprocessing time.
However, for a chain-like DAG whose construction time takes $O(n^{2})$, the
computation will be quadratic as well.
In fact, construction of badly behaving structures should be prevented upfront.
\newline It is possible to compute the longest query after every insertion,
causing an extra cost of $O(i\log{i})$ for the $i$th curve, which will cause
an overall deceleration in the preprocessing stage.
\newline The question of when to compute this value remains, since there is
clearly a tradeoff between the cost of computation and the increasing length
of the structure.
If the longest query path length is computed often, it is more likely to
exceed the expected preprocessing bound of $O(n\log{n})$ time to which the
algorithm aims.
If the computation is delayed, then large structures that behave badly will
 be discarded rather late after a costly construction and computation,
exceeding the expected preprocessing bound.
Another possibility is to maintain the longest query path length in the nodes
themselves, however this option increases the memory significantly since every
trapezoid node should keep an order of the number of different search paths
ending at it.
}

Another idea may be to simply observe the length of the path that was taken in order
to find the left endpoint of a newly inserted segment. Together with the fact that one also
ensures that the size of the DAG is linear it would be ensured that one build takes
$O(n\log n)$ time. However, the resulting DAG may even contain a search path of linear
size. Simply, take the recursive construction for the lower bound on~\depth/\lqpl
and extend all segments to the right. Clearly locating the left endpoints of newly inserted
segments takes $O(\log n)$ time. However, since they all extend to the right a search
path in this region is linear. But at least this triggers a solution for a static
subdivision.

\textbf{A solution for static subdivisions:}
The above algorithm ensures that the construction of a DAG takes $O(n\log n)$ time.
As this DAG may still contain paths of linear length we close the construction
by computing the longest path as discussed above. In case the DAG is ok, this
check takes $O(n\log n)$ otherwise we stop the computation and rebuild the DAG.
However, one of the key features of the data structure is that it supports dynamic
scene which includes insertions as well as deletions of segments. Thus, this solution
is obviously not satisfying.

\ignore{
\textbf{A solution for static subdivisions:} We cannot be sure
that computing~\lqpl would take expected $O(n\log{n})$ time\tcnote[Michal]{rephrase},
but we can verify that the expected construction time,
in case of static scenes, would not violate this bound.
We use similar insertion scheme, as described in section~\ref{sec:RIC},
with slightly modified conditions, promising that the expected
construction time would not exceed $O(n\log{n})$.
We ensure that every insertion takes $O(\log{n})$, by counting the length
of the search path of the curve's left endpoint. If at some point it exceeds
the predefined bound, indicating that a long search path exists,
 then a rebuild will be issued.
We should also keep track of the size of the data structure
after every insertion, as described in section~\ref{sec:RIC}.
However, this does not mean that~\lqpl is valid.
A counter example is a scenario based on the recursive scheme
 mentioned earlier, with the property that all segments extend to the right.
For a top to bottom insertion order,
the preprocessing time would be fine. On the other hand,
a query in the rightmost bottom trapezoid would take $O(n)$.
Now we run the method that computes~\lqpl in
a DFS-like manner, described earlier.
However, we stop the process after $O(n\log{n})$ time and
in case no answer was returned during that time,
implying that~\lqpl violates the $O(\log{n})$ bound, we issue a rebuild.
In total, this solution requires expected $O(n\log{n})$
preprocessing time.
It would not work with dynamic subdivisions efficiently, since the
\lqpl computation may be held after every insertion, exceeding the
complexity bounds.
}

\ignore{
\textbf{A solution for static subdivisions:} We describe here an optional
construction algorithm for static subdivisions with an expected $O(n\log{n})$
preprocessing time that creates a data structure whose length of longest
query path is at most $O(\log{n})$.
Given a planar partition of $n$ non-intersecting curves, use the random
incremental insertion in order to build the search structure, as described
in~\cite{Mulmuley1990_fast_planar_part_alg,Seidel1991_simp_fast_inc_rand_td}.
Similarly to the variation elaborated in Section~\ref{sec:impl_details},
defined in~\cite{CG-alg-app}, we will perform the following tests:
\vspace{-2mm}
\begin{itemize}
\item During every insertion the number of nodes in the path to the trapezoid
containing the left endpoint should be counted.
    If at some point it exceeds $c_1\log{n}$ for a suitable constant $c_1$,
the insertion is stopped and a rebuild will be issued, since it implies that
there is a valid search path whose length is too large.
\item After every insertion one must verify that the size of the data structure
does not exceed $c_2n$, for a different constant $c_2$.
\end{itemize}
\vspace{-2mm}
These two conditions promise that the expected construction time would not
exceed $O(n\log{n})$.
However, they do not give any bound on the~\lqpl.
An example is a scenario similar to Figure~\ref{fig:depth_vs_log_qlen} with the
property that every segment extends to the right boundary. For a top to bottom
 insertion order every insertion will take $O(\log{n})$ time. A query in the
right bottom trapezoid will take $O(n)$. Thus the~\lqpl need to be checked as
 well.

It is possible to run the method that computes the length of the longest query
path in a Depth-First-Search like traversal, as described earlier.
We run it for $c_3n\log{n}$ steps, for a certain constant $c_3$.
If it returned with an answer than the length of the longest query path is valid.
Otherwise, the longest query path length violates the allowed bounds, and a
rebuild should be issued.
The expected running time of this method is $O(n\log{n})$, however, since we
stop it after $O(n\log{n})$ steps the overall running time of deciding whether
the longest query is too long is $O(n\log{n})$.
This implies that that total preprocessing time of this construction os a \PL
search structure is expected $O(n\log{n})$.
This solution would not work with dynamic subdivisions efficiently, since the
decision method may be called after every insertion exceeding the insertion
complexity.
}
}

\ignore{
\secspace\section{A Discussion On The DAG Depth vs The Longest Query Path}\secspace

\label{sec:depth_vs_longest_query}
Reviewing the literature~\cite{CG-alg-app} about the construction of a search
structure that will not violate the complexity bounds, raised the discussion
on the choice of DAG depth over longest query path length, in order to keep
track of the maximum-length query in the constructed DAG.
This replacement seems minor, at first glance.
However, certain DAGs can have a significant difference between these two values.
Figures~\ref{fig:sqrt_query_linear_depth},~\ref{fig:depth_vs_log_qlen} illustrate
this.

\begin{wrapfigure}{r}{0.3\textwidth}
  \vspace{-30pt}
  \begin{center}
    \includegraphics[width=0.3\textwidth]{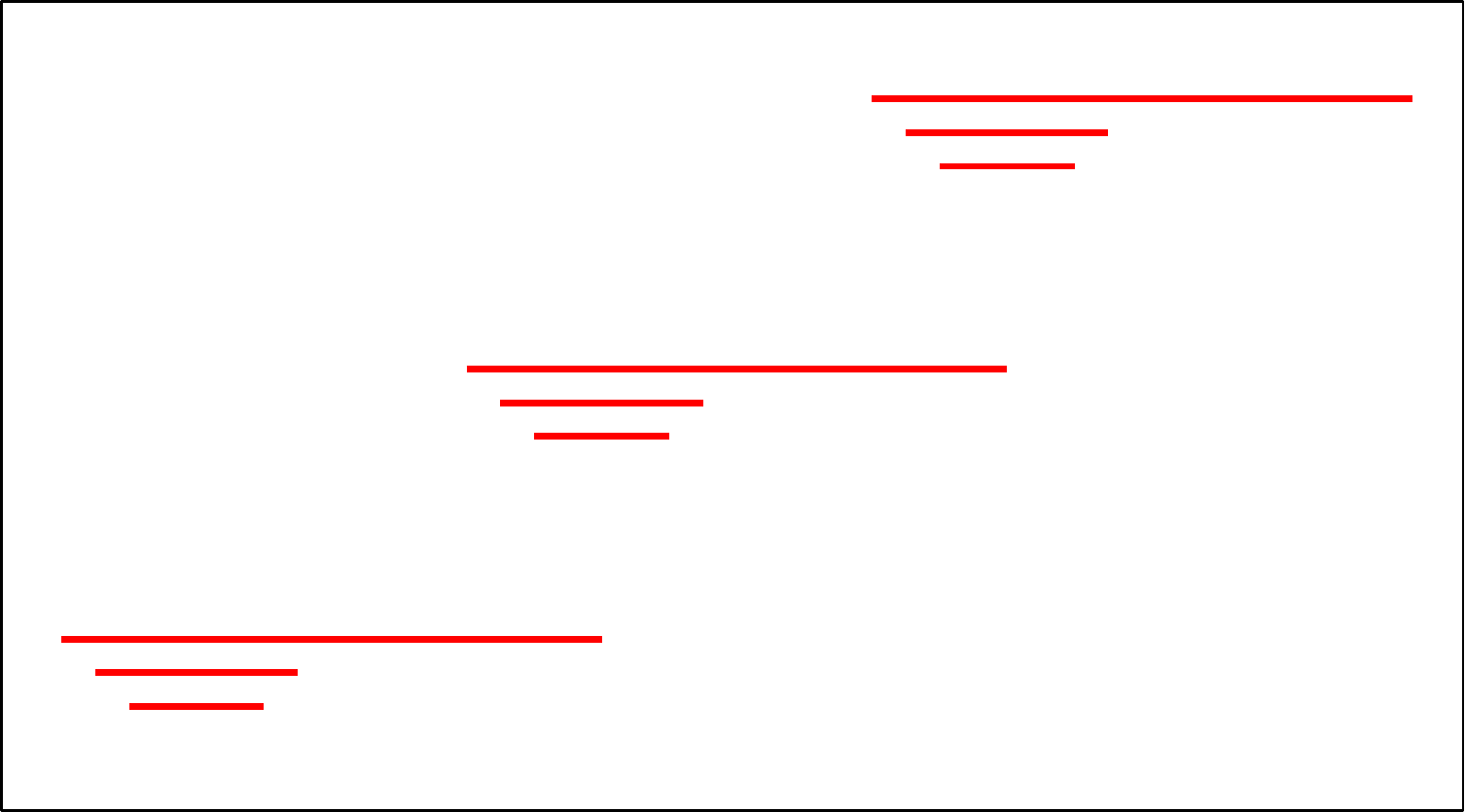}
  \end{center}
  \vspace{-15pt}
  \caption{}
  \vspace{-25pt}
  \label{fig:sqrt_query_linear_depth}
\end{wrapfigure}

\begin{figure*}[!htp]
    \centering
    \setlength{\tabcolsep}{3pt}
    \begin{tabular}{c}
        \begin{tabular*}{1\textwidth}{c c}
            \subfloat{\includegraphics[width=0.5\textwidth]{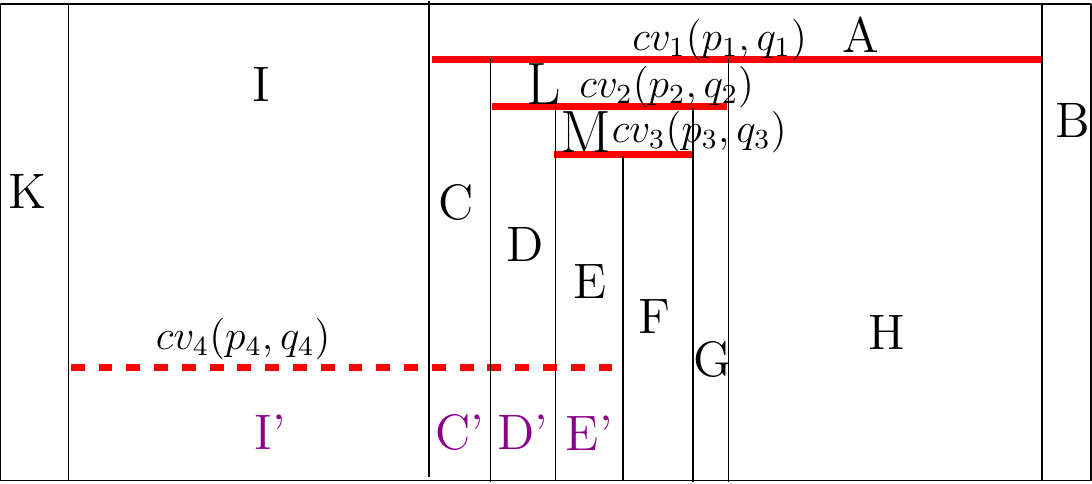}} &
            \subfloat{\includegraphics[width=0.5\textwidth]{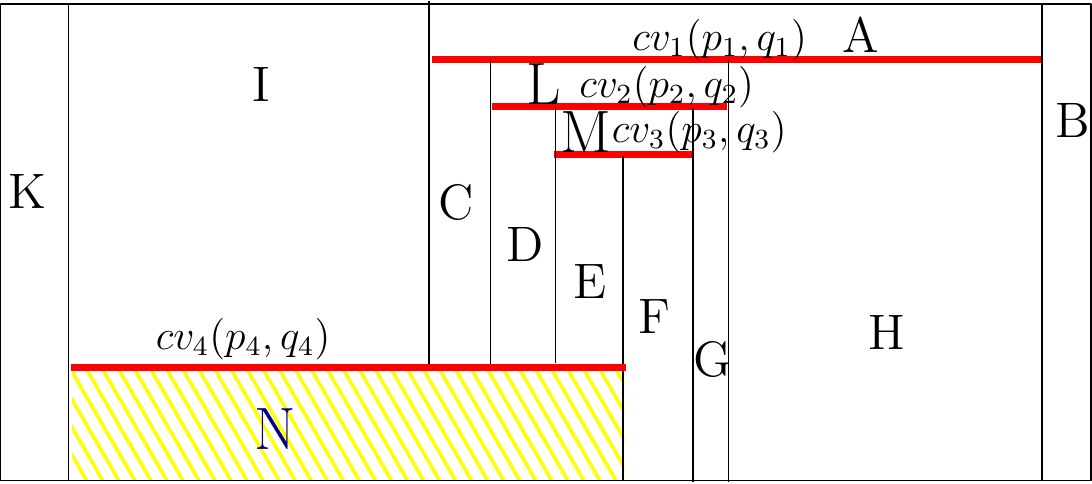}} \\
            \sf{(a) Trapezoidal map before merging} &
            \sf{(c) Trapezoidal map after merging}
        \end{tabular*} \\
        \vspace{2mm}
        \begin{tabular*}{1\textwidth}{c c}
            \subfloat{\includegraphics[width=0.5\textwidth]{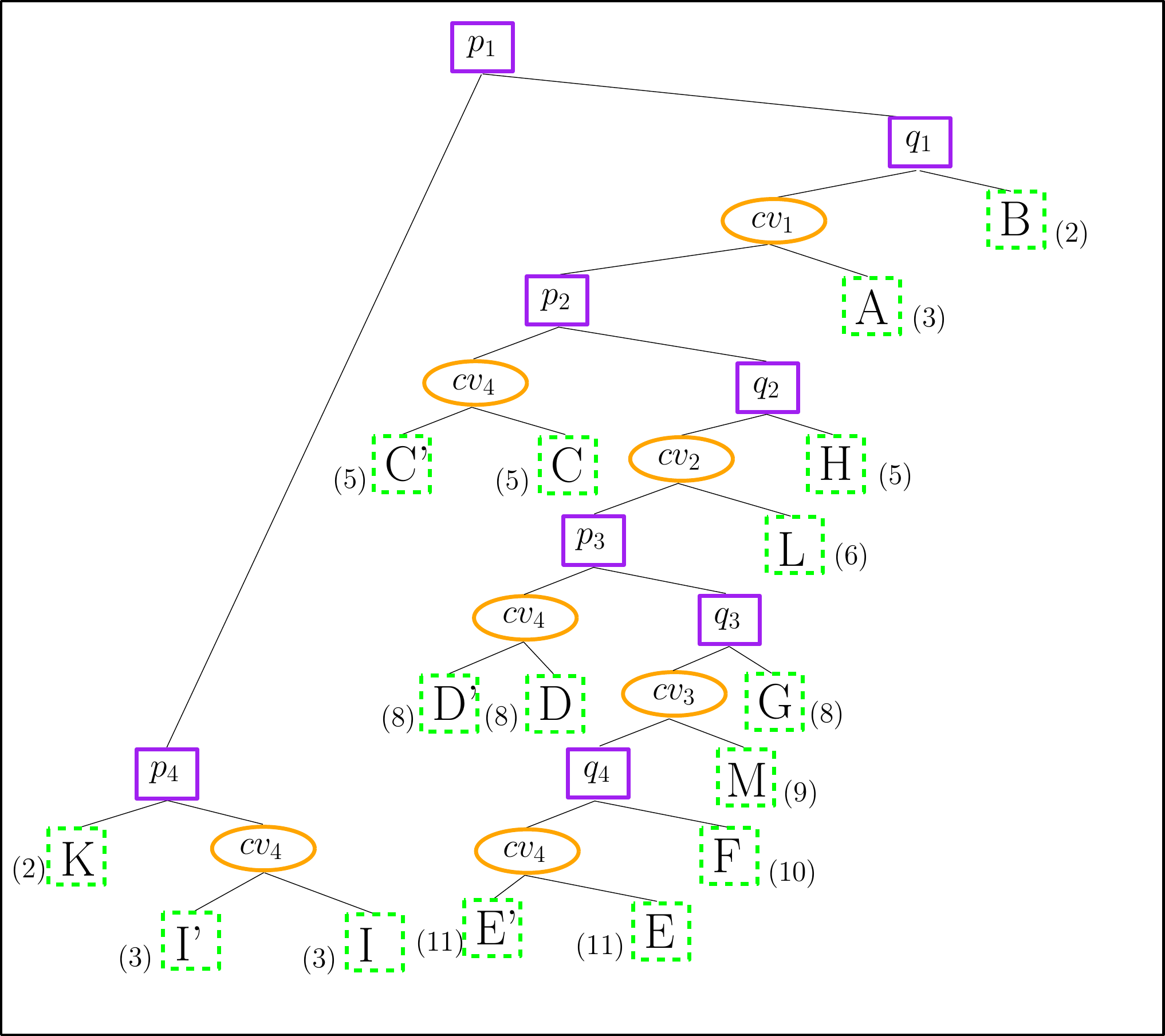}} &
            \subfloat{\includegraphics[width=0.5\textwidth]{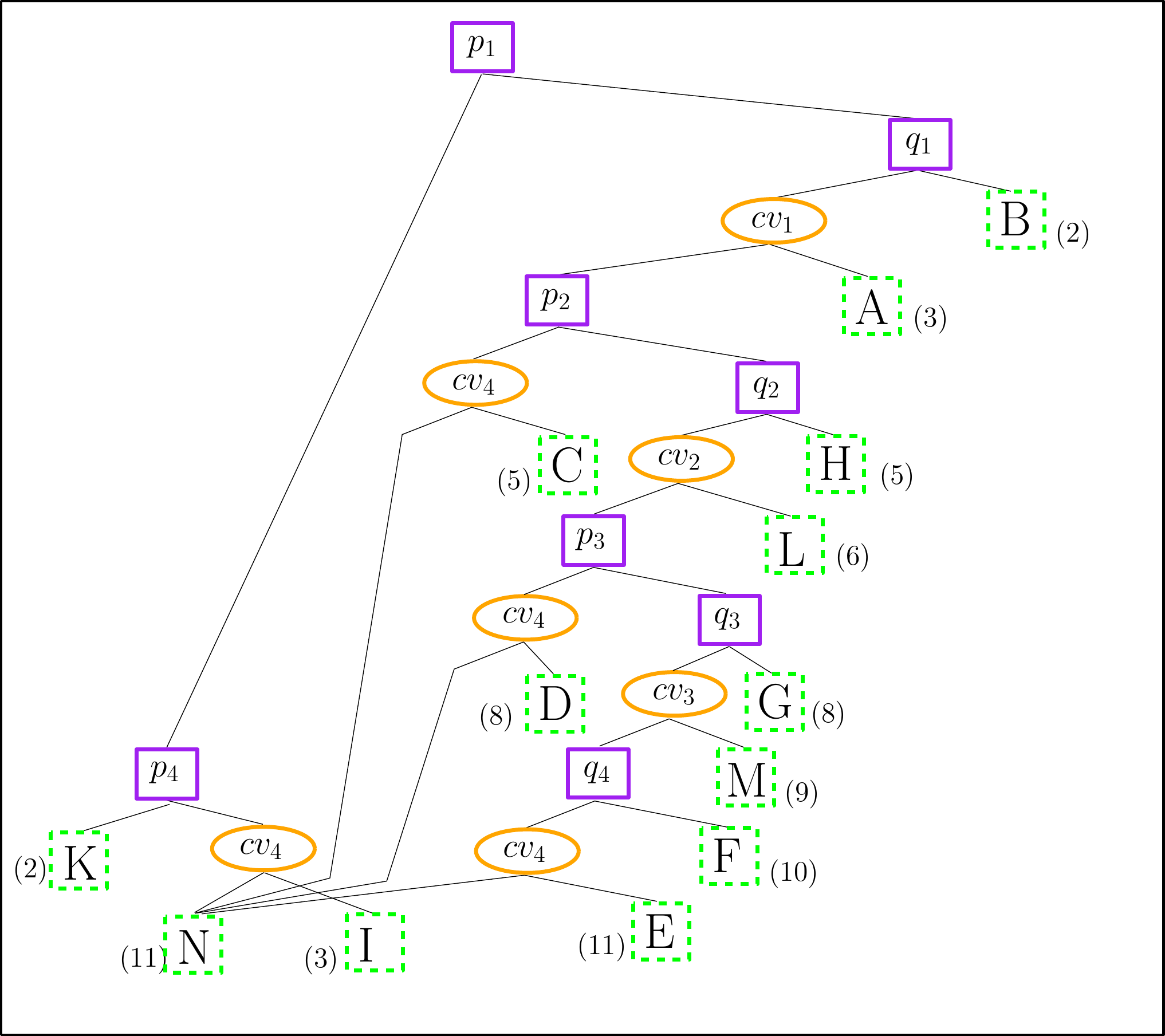}} \\
            \sf{(b) The DAG before merging} &
            \sf{(d) The DAG after merging}
        \end{tabular*}
    \end{tabular}
    \caption{\sf{The insertion of $cv_4$. (a),(b) display the Trapezoidal map and the created DAG, respectively, after the insertion of $cv_4$ but before the merges are performed. (c),(d) display the Trapezoidal map and the created DAG, respectively, after the required merges are performed. In (b),(d) the depth of the generated trapezoids is specified in brackets. Obviously the insertion of $cv_4$ causes merges that enlarge the depth of the DAG. See the change in the depth of trapezoid $N$ in (d) comparing to its matching trapezoid $I'$ in (b), for instance.}}
    \label{fig:sqrt_query_linear_depth_before_after}
    \vspace{-25pt}
\end{figure*}


Figure~\ref{fig:sqrt_query_linear_depth} presents a setting of input segments
with a possible large difference between the depth and longest query values.
Inserting the input from top to bottom generates a DAG with $O(n)$ depth and
$O(\sqrt(n))$ longest query path length.
The input segments are divided into $\sqrt{n}$ blocks of $\sqrt{n}$ segments each.
In each block the longest query is $O(\sqrt{n})$.
Figure~\ref{fig:sqrt_query_linear_depth_before_after} illustrates the
Trapezoidal Decomposition and the matching DAG of a setting of 9 curves,
after the insertion of the $4$th segment ($cv_4$), which is the first
segment of the second block. The segments of the first block
($cv_1$, $cv_2$, $cv_3$) were already inserted. Figure (a) shows the
Trapezoidal map before the required merges and (c) shows the map after
the merges are over.
Figures~\ref{fig:sqrt_query_linear_depth}(b),(d) illustrate the matching
DAGs for (a),(c), respectively.
The top segment of each block ($cv_4$ in
Figure~\ref{fig:sqrt_query_linear_depth_before_after}) splits trapezoids
that were created by the previous block.
For every trapezoid in the top-most block, whose query path does not pass
through $cv_4$, the depth and the longest query path length are the same.
New trapezoids that were created by the split($C'$,$E'$,$D'$), having $cv_4$
as their top curve and the bottom boundary as their bottom curve, are merged
(in (c), trapezoid $N$ is the result the merge of $I'$,$C'$,$E'$,$D'$ in (a)).
All trapezoids to the left of $p_1$ below $cv_4$ will now have a larger depth
 than their query paths.
For example, in (c), For a query in the left part of trapezoid $N$ the depth
 is 11 but the longest query path would go from the root to the left child,
skipping the sub-graph dealing with the first block of segments, having a
length of 3.
Now, after inserting the $O(\sqrt{n})$ blocks, the deepest leaf in the graph
is the trapezoid below $cv_n$. Its depth is $O(n)$, which is in fact the
depth of the DAG.
A path of a query point located in this trapezoid will start from the root
$p_1$ and continue along the leftmost path passing $\sqrt{n}$ left-right
decision nodes, going right in the last node, accessing the sub-graph
representing the last block, reaching the required trapezoid in $O(\sqrt{n})$ steps.
The total length of the query path would be $O(\sqrt{n})$, as $O(\sqrt{n})$
nodes were traversed on the leftmost path and $O(\sqrt{n})$ nodes were
traversed in the last block.
This bound is true for any query under this setting.

A careful computation shows that the above construction for $n$ segments
generates a DAG with $3n-\sqrt{n}+2$ depth and $4\sqrt{n}$ length of the
longest query path.



\begin{wrapfigure}{r}{0.5\textwidth}
  \vspace{-25pt}
  \begin{center}
    \includegraphics[width=0.5\textwidth]{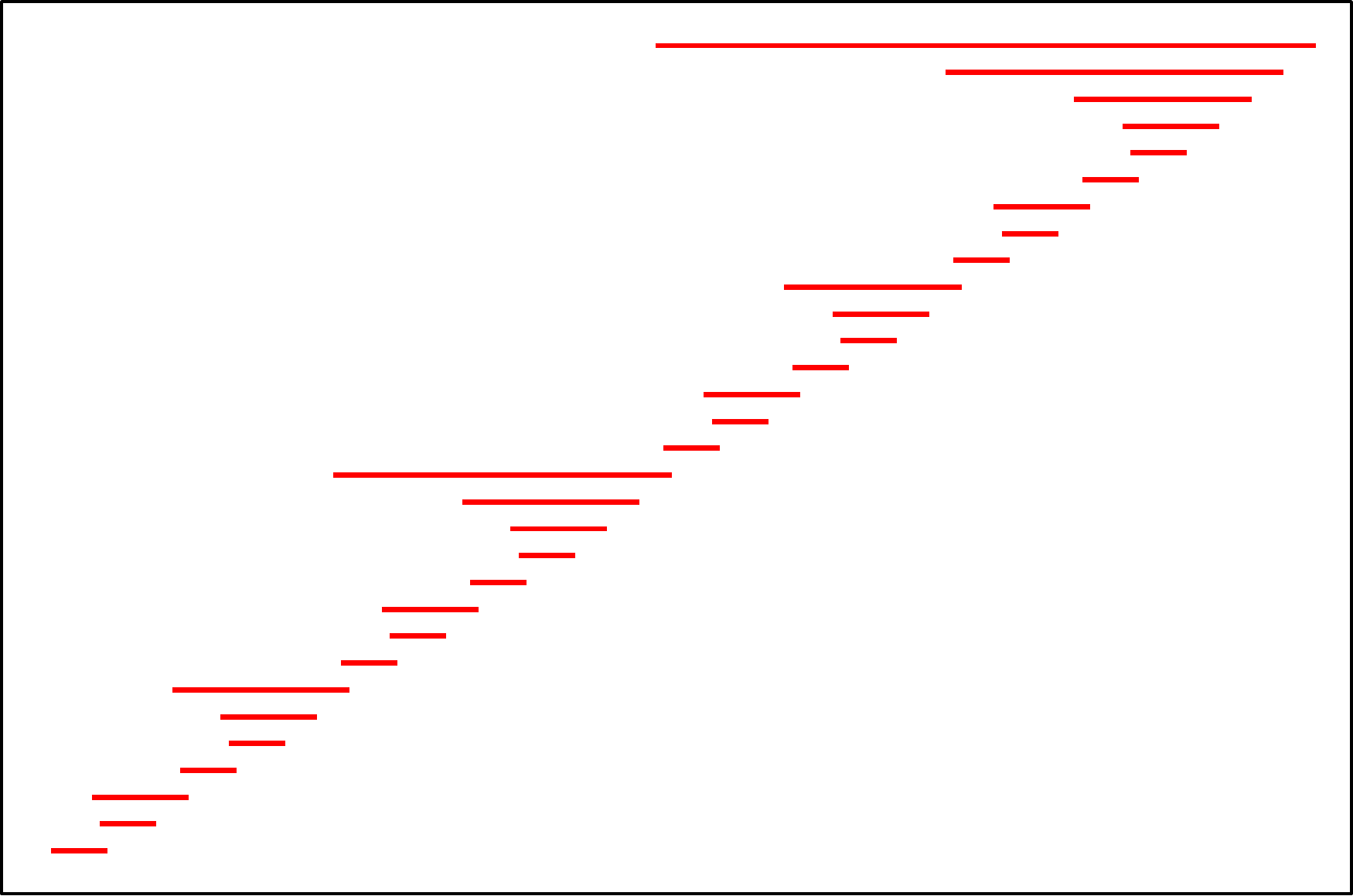}
  \end{center}
  \vspace{-10pt}
  \caption{\sf{For a top-to-bottom insertion order the resulting DAG will have a linear depth, whereas the longest query path will be $O(\log(n))$}}
  \vspace{-15pt}
  \label{fig:depth_vs_log_qlen}
\end{wrapfigure}

One should notice that our implementation, as described in
section~\ref{sec:impl_details}, which uses the DAG depth instead of the
longest query path length, will discard the DAG built for such an insertion order.
However, in such a setting, the structure should also be discarded when using
 the longest query path length, since this value exceed $O(\log{n})$.
A better example would be the recursive scenario presented in
Figure~\ref{fig:depth_vs_log_qlen}. The constructed DAG, for a top-to-bottom
insertion order, will have a linear depth (for the trapezoid below the
leftmost segment). However, any query path length would be $O(\log{n})$.
As opposed to the previous example, the generated DAG will be discarded
by our implementation, even though the length of any query path would
never exceed $O(\log{n})$.

\begin{proposition}
The random incremental construction of a Trapezoidal Decomposition of a
planar partition of $n$ non intersecting curves can generate a \PL search
structure, a DAG,
whose depth is $O(n)$ and longest search path is only $O(\log{n})$.
Such a DAG will be considered invalid when checking the structure conditions,
if the depth is used instead of the length of longest query path.
\end{proposition}

The fact that these two values may differ is due to the nature of a DAG to
have more than one path to a node.
A DAG can contain several paths to a leaf, most of cannot be realized by any
 query path.


The incentive to replace the longest query path length with the DAG depth
is related to the complexity of computing the longest query path.
The DAG depth is an independent external property of the graph, easily
maintained after every structural update at cost proportional to the cost
of the update itself.
The longest query path, however, is a non-trivial property depending on the
underlying geometry, and in its computation geometric comparisons are involved.

Let $T$ be the matching trapezoidal search tree to a given DAG $G$.
$T$ can be produced by a similar insertion order without any merges.
Each inserted curve will add two vertical extensions on its endpoints,
extending until an already inserted curve is met.
Clearly, all DAGs represent a single decomposition of the plane, while
different insertion orders will generate different trees whose underlying
partitions are not the same.
A leaf in $G$ will be a leaf in $T$ as well, if none of the endpoints
of curves inserted after the top and bottom curves (if exist) of the
underlying trapezoid fall inside the region defined by it.
Otherwise, it will be a combination of leafs in $T$, each defining a
vertical slab of the underlying trapezoid of the leaf in $G$.
\tcnote[Michal]{Verify that this is correct!!!!!}

One can choose to use a trapezoidal search tree, as defined above.
This search structure will have a single path from root to any leaf,
which is the exact query path to the underlying trapezoid.
In other words, in such a tree the longest query path will set the depth,
hence the depth can be easily accessed.
However, as showed in~\cite{SeidelA2000_wc_query_comp}, the trapezoidal
search tree uses $\Omega(n\log(n))$ space, but by using cuttings in order
to create a DAG, the space can be reduced to linear.
\tcnote[Michal]{Verify that their trapezoidal search tree is the tree that
we describe here. do we know that the expected size is $O(n\log(n))$?,
and understand how much influence has the log factor}.
Table~\ref{tbl:tree_dag_size} displays the difference in the size of the
two structures.

\begin{table}[t]
  \vspace{-20pt}
  \caption{
    Number of tree nodes vs. number of DAG nodes for the 
    same input with the same insertion order.}
  \begin{center}
    \begin{tabular}{| c | c | c | c |}
    \hline
    \# Arrangement Edges & \# Tree nodes & \# DAG nodes & ratio\\
    \hline
    138 & 1263 & 681 & 1.85\\
    \hline
    285 & 3167 & 1492 & 2.12\\
    \hline
    1483 & 23511 & 8019 & 2.93\\
    \hline
    2977 & 51551 & 16330 & 3.15\\
    \hline
    14975 & 350629 & 84576 & 4.14\\
    \hline
    29973 & 759075 & 169355 & 4.48\\
    \hline    \end{tabular}
  \end{center}
  \label{tbl:tree_dag_size}
  \vspace{-20pt}
\end{table} 

If a DAG is indeed chosen as a search structure, how costly would computing
the longest query path length be?
If the vertical lines are extended to infinity, then they create $O(n^{2})$
different regions.
A na\"{i}ve algorithm would compute the length of a query in each region and
return the maximal one.
Another algorithm, which as far as we know is the most efficient, would go
over the DAG and mimic the trapezoidal search tree.
It begins at the root, and recursively checks the two children, taking the
(maximal + 1).
It goes on recursively until the slabs are found.\tcnote[Michal]{I think it
is unclear} Each node on the way contributes to the accumulated length.
The complexity of this algorithm is expected $O(n\log(n))$, and is correlated
to the structure of the traversed DAG.
One may only compute the value in the end of the construction in order to
make sure the DAG does not violate the requirements.
This would not effect the preprocessing time. However, for a chain-like DAG
the computation will take $O(n^{2})$. This DAG is considered bad and might
have been discarded much earlier.
In fact one should prevent the construction of wrong DAGs upfront.
The longest query can be computed after every insertion, causing an extra cost
of $O(i\log{i})$ for the $i$th curve, which will cause an overall deceleration
in the preprocessing stage.
The question of when to compute this value remains, since there is clearly
a trade off between the cost of computation and the increasing length of the
structure.
If the longest query path length is computed often, it is more likely to
exceed the expected preprocessing bound of $O(n\log{n})$ time to which the
algorithm aims.
If the computation is delayed, then large structures that behave badly will
be discarded rather late after a costly computation, exceeding the expected
preprocessing bound.

Another possibility is to maintain the longest query path length in the nodes
themselves, however this option increases the memory significantly since every
trapezoid node should keep an order of the number of slabs contained in
it.\tcnote[Michal]{number of leafs in the matching tree composing the leaf
in the graph G. Each leaf corresponds to a single query path to the trapezoid}

Thus, using the depth which is an upper bound of the longest query path can help.
The risk involved in this replacement, is that one would initiate a rebuild
based on the large depth, while the longest query path is rather short.
(An example is the search structure for the scenario in
figure~\ref{fig:depth_vs_log_qlen}).
This may lead to more rebuilds than necessary, making the preprocessing
time violate the expected limit.



As shown in~\cite{CG-alg-app}, the expected longest query path length is
$O(\log{n})$.
We conjecture that the expected DAG depth is logarithmic as well.
If this conjecture is proved to be true, it would support the use of depth
instead of longest query path length, since the expected depth-query ratio
would be $O(1)$.
Our experiments imply that these values are indeed close.

\begin{table}[b]
  \begin{center}
    \begin{tabular}{| c | c | c | c |}
    \hline
    \# arrangement edges &~\depth &~\lqpl &~\depth/\lqpl \\
    \hline
    43 & 17 &	17&	1\\
    \hline
    51	& 17	& 16	& 1.06 \\
    \hline
    53	& 20 &	20 & 1 \\
    \hline
    1940& 35& 34 & 1.03 \\
    \hline
    2340 & 40 & 40 &1 \\
    \hline
    56142 &	61	& 57	& 1.07\\
    \hline
    56556 &	60 &	55 &	1.09 \\
    \hline
    59920	& 61	& 56	& 1.08 \\
    \hline
    215114	& 71	& 70	& 1.01 \\
    \hline
    231652	& 79	& 74	& 1.03 \\
    \hline
    237908 &	68 &	66 &	1.03 \\
    \hline
    \end{tabular}
  \end{center}
  \caption{\depth/\lqpl ratio in search structures for arrangements of non-intersecting random segments (random permutation, no rebuilds)}
  \label{tbl:rand_segs_ratio}
\end{table} 

\begin{table}[!htp]
  \begin{center}
    \begin{tabular}{| c | c | c | c | c |}
    \hline
    \# Voronoi sites & \# arrangement edges & Average~\depth/\lqpl \\
    \hline
    300 &	881 &	1.024 \\
    \hline
    1000 &	2977 &	1.022 \\
    \hline
    10000 &	29973 &	1.058 \\
    \hline
    100000	& 299968 &	1.078 \\
    \hline
    \end{tabular}
  \end{center}
  \caption{Average~\depth/\lqpl ratio in a search structure for a Voronoi diagram of random sites with random permutation of the Voronoi edges}
  \label{tbl:rand_sites_ratio}
\end{table} 

\begin{table}[!htp]
  \begin{center}
    \begin{tabular}{| c | c | c | c | c |}
    \hline
    \# Voronoi sites & \# arrangement edges & Average~\depth/\lqpl \\
    \hline
    300 &	598 &	1.006 \\
    \hline
    1000 &	1998 &	1 \\
    \hline
    10000 &	22812 &	1.07 \\
    \hline
    100000 &	287214 &	1.144 \\
    \hline
    \end{tabular}
  \end{center}
  \caption{Average~\depth/\lqpl ratio in a search structure for a Voronoi diagram of random sites on a circle and an origin with random permutation of the Voronoi edges}
  \label{tbl:circle_sites_ratio}
\end{table}

\subsecspacea\subsection{A Solution for Static \PL}\subsecspaceb

We describe here an optional construction algorithm with an expected
$O(n\log{n})$ preprocessing time that creates a data structure whose
length of longest query path is at most $O(\log{n})$.
Given a planar partition of $n$ non-intersecting curves, use the random
incremental insertion in order to build the search structure, as described
in~\cite{Mulmuley1990_fast_planar_part_alg},~\cite{Seidel1991_simp_fast_inc_rand_td}.
Similarly to the variation elaborated in Section~\ref{sec:impl_details},
defined in~\cite{CG-alg-app}, we will perform the following tests:
During every insertion the length of the path to the trapezoid into which
the left endpoint falls should be measured.
If it exceeds $c_1\log{n}$ for a suitable constant $c_1$, a rebuild will
be issued, since it implies that there is a valid search path whose length
is large enough.
After every insertion one must verify that the size of the data structure
does not exceed $c_2n$, for a different constant $c_2$.

These two conditions promise that the expected construction time would not
exceed $O(n\log{n})$.

[-- It is not enough to just verify this , an example will be the recursive
extended to the right --]

Now, it is possible to run the method that computes the length of the longest
query path while mimicking the Trapezoidal Search Tree, as described earlier.
We run it for $c_3n\log{n}$ steps, for a certain constant $c_3$.
If it returned with an answer than the length of the longest query path is valid.
Otherwise, the longest query path length violates the allowed bounds, and a
rebuild should be issued.
The expected running time of this method is $O(n\log{n})$, however, since we
stop it after $O(n\log{n})$ steps \tcnote[Michal]{maybe use time instead of steps?},
the overall running time of deciding if the longest query is too long is
still $O(n\log{n})$.

This implies that that total preprocessing time of this construction os a
\PL search structure is expected $O(n\log{n})$.

[-- this solution will not work for dynamic --]
} 
\secspace\section{Nearest Neighbor Search in Guaranteed $O(\log n)$ Time}\secspace
\label{sec:nearest neighbor}
\label{sec:nn}

As stated in the Introduction, we were challenged by the claim of
Birn et al.~\cite{SANDERS2010_simp_fast_nn} that guaranteed logarithmic
nearest-neighbor search can be achieved via efficient point location
on top of the Voronoi Diagram of the input points, but that this approach
{\em ``does not seem to be used in practice"}.
With this section we would like to emphasize that such an approach is available
and that it should be considered for use in practice.
Using the RIC planar point location, the main advantage would be
that query times are stable and independent of the actual scenario.

\subsecspacea\subsection{Nearest Neighbor Search via Voronoi Diagram}\subsecspaceb
\label{subsec:guaranteed_nn}

Given a set $P$ of $n$ points, which we wish to
preprocess for efficient point location queries,
we first create a Delaunay triangulation (DT) which takes $O(n\log n)$
expected time. The Voronoi diagram (VD) is then obtained by dualizing.
Using a sweep, the arrangement representing the VD,
which has at most $3n-6$ edges, can be constructed in $O(n\log{n})$ time.
However, taking advantage of the spatial coherence of the edges,
we use a more efficient method that directly inserts VD edges while
crawling over the DT.
The resulting arrangement is then further processed by our RIC implementation.
If Conjecture~\ref{con:depth} is true then this takes expected
$O(n\log{n})$ time.
Alternatively, it would have been possible to implement the
solution presented in Subsection~\ref{ssec:static_sub_sol}
(for which we can prove expected $O(n\log{n})$ preprocessing time).

\subsecspacea\subsection{Nearest Neighbor Search via Full Delaunay Hierarchy}\subsecspaceb

The full Delaunay hierarchy  (FDH) presented in~\cite{SANDERS2010_simp_fast_nn}
is based on the fact that one can find the nearest neighbor by performing a
greedy walk on the edges of the Delaunay triangulation (DT).
The difference is that the FDH keeps all edges that appear during the randomized
construction~\cite{AmentaCR2003_inc_construct_con_brio} of the DT
in a flattened $n$-level hierarchy structure,
where level $i$ contains the DT of the first $i$ points.
Thus, a walk that starts at the first point is accelerated due to long
edges that appeared at an early stage of the construction process
while the DT was still sparse.
The FDH is a very light, easy to implement, and fast 
data structure with expected $O(n\log{n})$ construction time
that achieves an expected $O(\log{n})$ query path length.
However, a query may take $O(n)$ time since the degree of nodes
can be linear.
\new For the experiments we used two  exact variants: a basic exact version (EFDH) and a (usually faster)  version (FFDH) that first performs a walk using inexact floating point  arithmetic and then continues with an exact walk.



\subsecspacea\subsection{Experimental Results}\subsecspaceb
\label{subsec:nn_experimental_results}

We compared our implementation for \NN search using the RIC \PL on the
Voronoi-diagram (\NAME) to the following exact methods:
EFDH, FFDH,
\CGAL's Delaunay hierarchy (CGAL\_DH)~\cite{Devillers2002_del_hierarchy},
and \CGAL's kd-tree (CGAL\_KD).%
\footnote{Due to similar performance we elided the kd-tree implementation
in ANN~\cite{MountA1997_ann}.}

\new All experiments have been executed on a Intel(R) Core(TM) i5 CPU M 450
with 2.40GHz, 512 kB cache and 4GB RAM memory,
running Ubuntu 10.10.
Programs were compiled using \texttt{g++} version  4.4.5
optimized with \texttt{-O3} and \texttt{-DNDEBUG}.
The left plot of Figure~\ref{fig:nn} displays the total query time
in a random scenario, in which both input points and query points are
randomly chosen within the unit square.
Clearly, all methods have logarithmic query time, however due to larger constants
\NAME is slower.
The other plot presents a combined scenario of
$(n-\lfloor\log{n}\rfloor)$ equally spaced input points on the unit circle
and $\lfloor\log{n}\rfloor$ random outliers. The queries are random points
in the same region.
In this experiment the CGAL\_KD and \NAME are significantly faster and maintain a
stable query time.
A similar scenario that was tested contains equally spaced input points on a circle
and a point in the center with random query points inside the circle.
The differences there are even more significant than in the previous scenario.
As for the preprocessing time in all tested scenarios,
obviously \NAME cannot compete with the fast construction time of the other methods.

\begin{figure}[t]
\vspace{-15pt}
\includegraphics[width=0.48\textwidth]{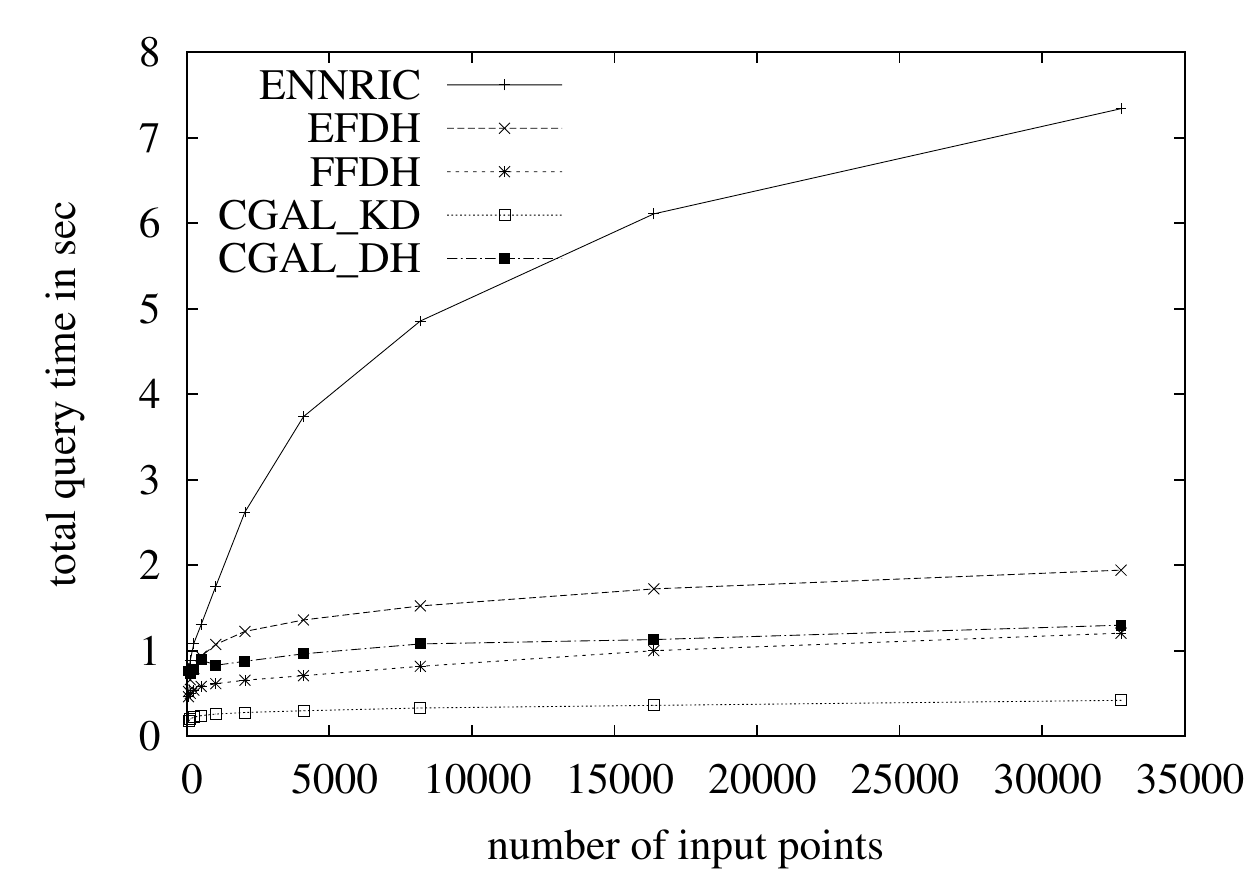}
\includegraphics[width=0.48\textwidth]{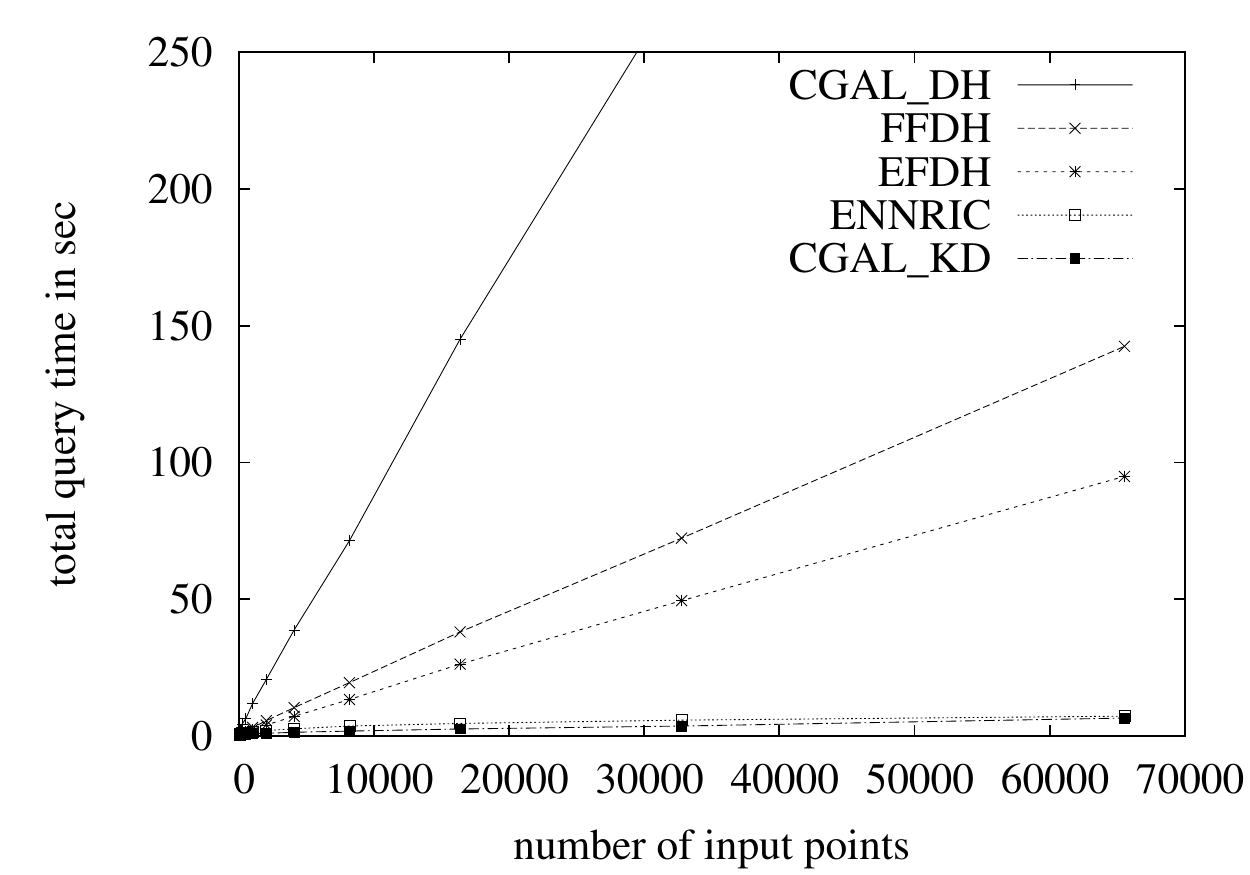}
\caption{
Performance of 500k \NN queries for different methods on two scenarios:
(left) random points; (right) circle with outliers. }
\label{fig:nn}
\vspace{-15pt}
\end{figure} 
\secspace\section{\protect\CGAL's New RIC Point Location}\secspace
\label{sec:impl_details}

With this article we announce our revamp of \CGAL's implementation
of planar point location via the randomized incremental construction
of the trapezoidal map,
which is going to be available in the upcoming \CGAL release 4.1.

Like the previous implementation by Oren Nechushtan~\cite{FlatoHHNE2000_planar_maps_cgal},
it is part of the ``2D Arrangements'' package~\cite{cgal:wfzh-a2-12} of \CGAL.
It allows both insertions and deletions of edges.
The implementation is exact and covers all degenerate cases.
Following the {\em generic-programming paradigm}~\cite{a-gps-99} it can be
easily applied to linear geometry but also to non-linear geometry such as
algebraic curves or B\'ezier curves.
The main new feature, and this is what triggered this major revision, is the support
for unbounded curves, as it was introduced for the ``2D Arrangements'' package
in~\cite{bfhmw-scmtd-07}, enabling point location on
two-dimensional parametric surfaces (e.g., spheres, tori, etc.) as well.

In addition we did a major overhaul of the code basis. In particular, we maintain
the depth~\depth of the DAG as described in
Section~\ref{sec:depth_vs_longest_query} such that~\depth is accessible
in constant time.
Thus we can now guarantee logarithmic query time after every operation.
Moreover, the data structure now operates directly on the entities of
the arrangement. In particular, it avoids copying of geometric data which
can significantly reduce the amount of additional memory that is
used by the search structure. This is important, since due to the
generic nature of the code it is not clear whether the
geometric types (user provided) are  referenced.

To the best of our knowledge,
this is the only available implementation of a point location method
with a guaranteed logarithmic query time
that can handle two-dimensional subdivisions to this generality.
Furthermore, it is the fastest available \PL method, in terms of query time,
for \CGAL arrangements.%
\footnote{A comparison to \CGAL Landmarks \PL~\cite{HaranH08}
is given in the Appendix~\ref{sec:landmarks}.}

\secspace\section{Open Problem}\secspace
\label{sec:conclusions}
\label{sec:further_work}
Prove Conjecture~\ref{con:depth}, that is, prove that it is possible
to rely on the depth~\depth of the DAG and still expect only a constant number
of rebuilds. This solution would not require any changes to the current implementation.\\

\noindent
\textbf{Acknowledgement:} The authors thank Sariel Har-Peled for sharing
Observation~1, which is essential to the expected
$O(n \log n)$ time algorithm for producing a worst-case linear-size and
logarithmic-time point-location data structure.

\michal{However, since we use a recursive algorithm for computing~\lqpl is it possible to find a tighter bound?}

\bibliography{bibliography}
\bibliographystyle{splncs}


\pagebreak
\newpage
\appendix
\section{Detailed Results of \depth/\lqpl Ratio Experiments}
\secspace
\label{sec:dl_ratio_results_appndx}

This appendix contains all experiments concerning the 
ratio of the depth \depth of a DAG and the length \lqpl of 
the longest query path in the same DAG. 
In addition to those that are mentioned in 
Section~\ref{ssub:dept_path_ratio_exp}
we also tested the special scenarios that we constructed in 
Section~\ref{ssub:depth_path_ratio} in order to achieve lower 
bounds on the worst case ratio of \depth and \lqpl. 
The set of segments is as depicted in Section~\ref{ssub:depth_path_ratio}.
However, in the experiments here we choose random order of insertion. 

In all experiments the two values hardly differ, that is, 
the largest ratio that we were able to observe was around $1.3$.
In the special scenarios, this value was even lower and in many 
cases \depth and \lqpl actually did not differ at all. 
However, for very large random scenarios, see Figure~\ref{fig:DL_rand_segs}, 
\depth was always a bit larger than \lqpl, but on the other hand
the largest observed ratio even went down to less than~$1.2$. 

This indicates that \depth and \lqpl behave sufficiently similarly. 
One can expect that an algorithm that 
rebuilds the DAG  as soon as \lqpl becomes larger than $c\log{n}$ 
would actually rebuild more often than an algorithm that rebuilds 
as soon as \depth becomes larger than  $1.3 c\log{n}$,
for some constant $c>0$.
This led us to venture Conjecture~\ref{con:depth}.

\subsection{Experiments}
\label{subsec:ratio-results-experiments}

We tested the ratio in the following scenarios:
\begin{enumerate}
  \item 
    Random line segments:
    Each segment was created from two random points in $[-1,1]^2$. 
    The number of generated segments was $\lfloor1.5^k\rfloor$,
    for $6\leq{k}\leq{19}$. The reported results are the average of 20 builds of 
    the search structure for the same random scenario. 
    See Figure~\ref{fig:DL_rand_segs} (left).

  \item 
    Voronoi diagram of random points: 
    For each scenario we took $2^k$ random sites for $6\leq{k}\leq{15}$.
    For each $k$ we generated $10$ different point sets and created 
    the search structure $7$ times, that is, the reported results are the 
    average of $70$ builds.
    See Figure~\ref{fig:DL_voronoi} (right).

  \item
    Lower bound construction for $O(\sqrt{n})$ ratio:
    We created the special scenarios according to the description in 
    Section~\ref{ssub:depth_path_ratio}, 
    each containing $k^2$ segments for $k\in\{10\cdot2^i|i\in\{1,\dots,6\}\}$.
    See Figure~\ref{fig:DL_sqrt} (left)
   
    \item
    Lower bound construction for $O(n/\log{n})$ ratio: 
    We created the special scenarios according to the description in 
    Section~\ref{ssub:depth_path_ratio}, 
    each containing $2^k$ segments for $8\leq{k}\leq{17}$.
    See Figure~\ref{fig:DL_sqrt} (right)
\end{enumerate}

\begin{figure}[h]
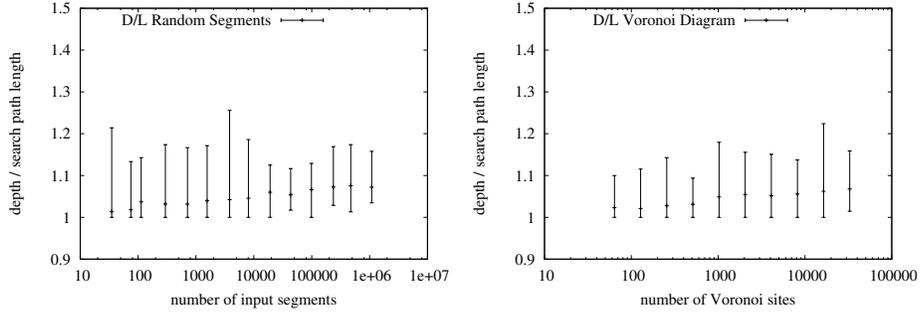

  \includegraphics[width=0.5\textwidth]{bench/DL_ratio/length_vs_depth_random_segments.pdf}
  \includegraphics[width=0.5\textwidth]{bench/DL_ratio/length_vs_depth_voronoi.pdf}
  \caption{
    \depth/\lqpl for arrangement of random segments (left) and Voronoi Diagram 
    of random sites (right). Plots show average value with error bars.
    }
  \label{fig:DL_voronoi} \label{fig:DL_rand_segs}
\end{figure}

\begin{figure}[h]
  \includegraphics[width=0.5\textwidth]{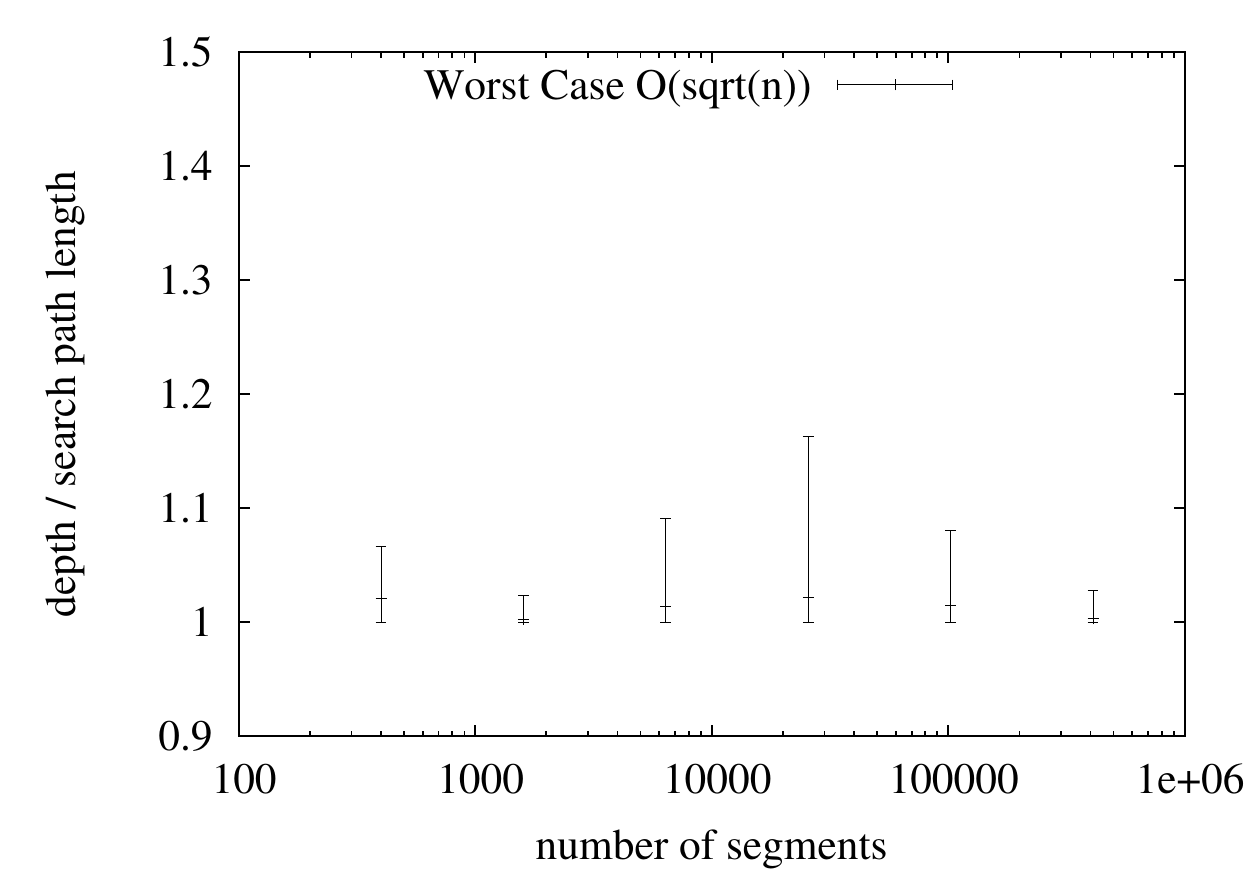}
  \includegraphics[width=0.5\textwidth]{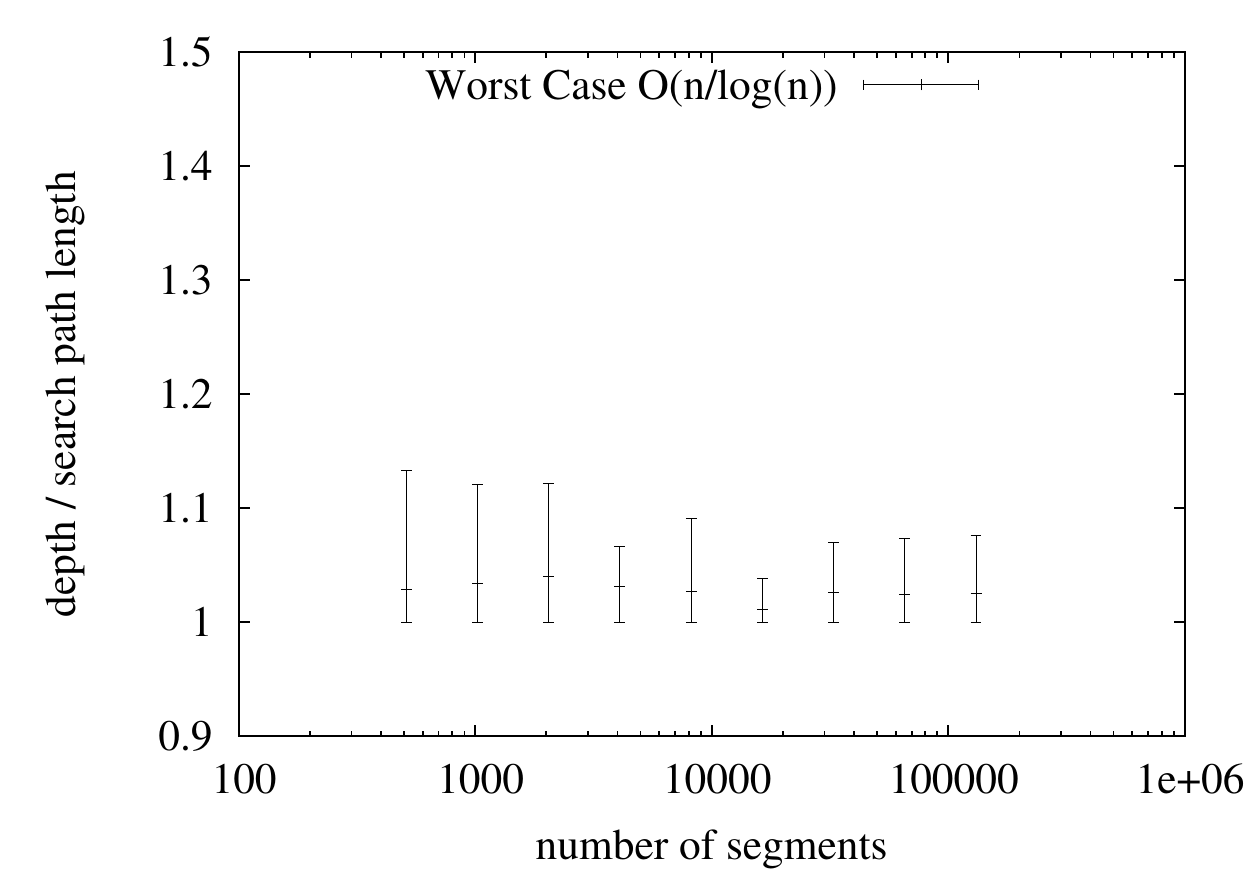}
  \caption{
    \depth/\lqpl for the example with worst case ratio $O(\sqrt{n})$ (left), 
    and $O(n/\log{n})$ (right). Plots show average value with error bars.
  }
  \label{fig:DL_sqrt} \label{fig:DL_n_logn}
\end{figure}

\clearpage
\section{Comparison to the \CGAL's Landmarks Point Location}
\secspace
\label{sec:landmarks}

We emphasize that the new 
implementation of the trapezoidal-map random incremental construction for
\PL (RIC) performs better than all other \PL methods available 
for \CGAL arrangements.

\begin{figure}[h]
  \centering
  \includegraphics[width=0.5\textwidth]{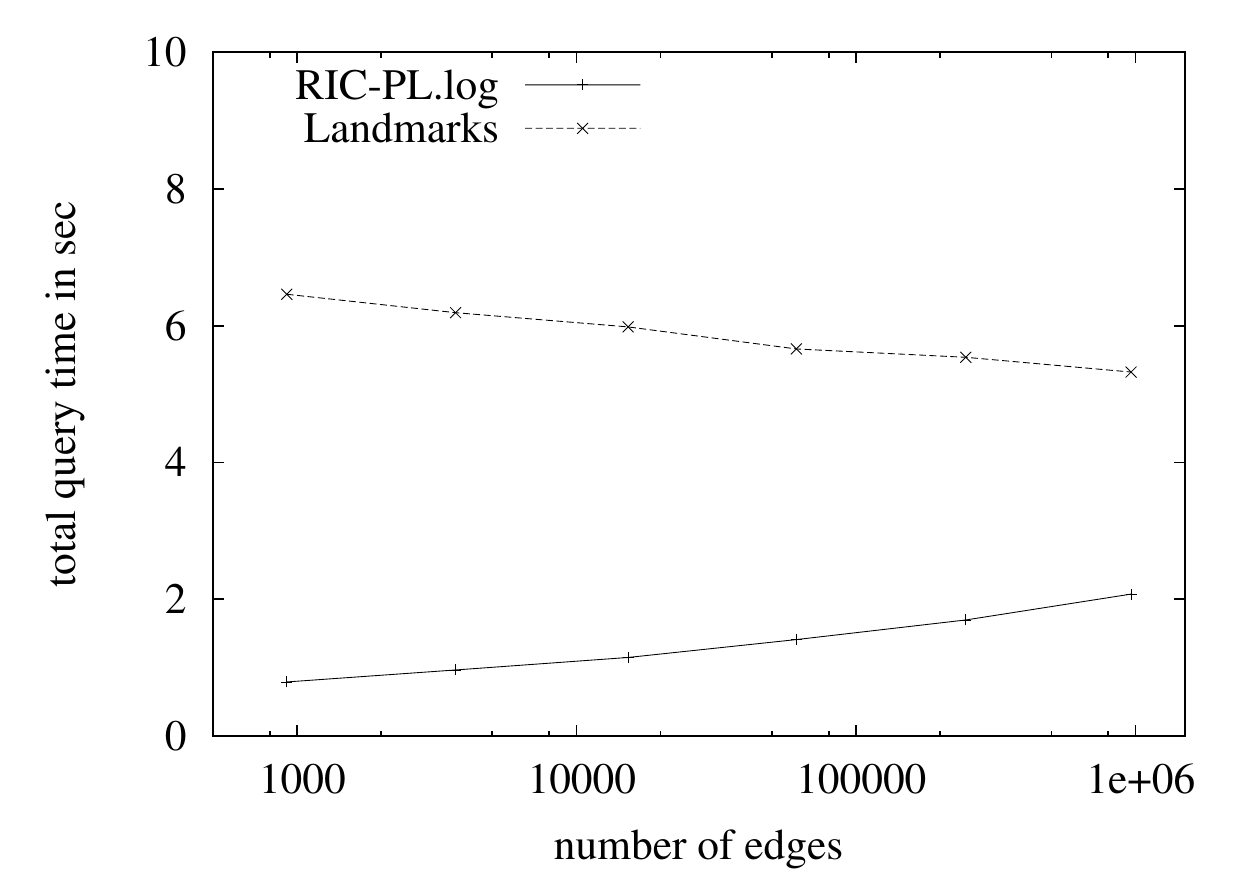}
  \caption{Comparing the total query time for 50k queries in random subdivision of a varying size using both the \CGAL Landmarks and the RIC \PL methods.}
  \label{fig:landmarks}
\end{figure}

Figure~\ref{fig:landmarks} displays the difference in the 
total query time in different arrangements of random segments
using the RIC vs. the Landmarks (LM) \PL.
The landmarks generator in this experiment created landmarks on a 
$\lceil\sqrt{V}\rceil \times \lceil\sqrt{V}\rceil$ grid 
($V$ is the number of vertices in the arrangement).
In~\cite{HaranH08} it is shown that 
for subdivisions of random segments the LM 
using the grid generator performs better than other 
\PL methods implemented in \CGAL,
other than the RIC.
As expected, the new RIC implementation outperforms the LM.
Obviously, the RIC query time is logarithmic. 
The slight improvement of the query time of the LM can be explained 
by the fact that, at some point,
while the number of input segments increases
the average complexity of a face decreases,
an effect that was for instance studied in~\cite{ahns-cofar-08}.


\clearpage
\secspace\section{Computing the Depth of ${\cal A}(T)$}\secspace
\label{sec:comp_arrangement_depth}
We would like to describe a linear space algorithm with $O(n\log n)$ runtime
for computing the depth of a collection of open trapezoids with the following properties:
their bases are $y$-axis parallel (vertical walls) and if
the top or bottom curves of two different trapezoids intersect
then the two curves overlap completely in their joined $x$-range.
The depth of such a collection is the maximum
number of trapezoids containing a common point,
that is, we are only interested in points located on faces of the arrangement of all trapezoids.
In Subsection~\ref{ssub:alg_arr_depth_rectangles}
we restate the algorithm of~\cite{as-cdaaa-13}
such that the general position assumption can be dropped.
The restated algorithm can handle rectangles
with independently open or closed boundaries,
and is more general than what we essentially need.
Subsection~\ref{ssec:reduction} defines a reduction from 
the collection of open trapezoids $T$ into a collection $R$ of
open axis-parallel rectangles such that the maximum depth in ${\cal A}(R)$
is the same as the maximum depth in ${\cal A}(T)$.
Finally, in Subsection~\ref{ssub:alg_arr_depth_general}
we describe a modification for the
restated algorithm such that it can compute the depth of 
the arrangement of all trapezoids 
created during the
construction of the DAG.

\subsecspacea\subsection{An Algorithm for Computing the Depth of a Collection
 of Axis-aligned Rectangles}\subsecspaceb
\label{ssub:alg_arr_depth_rectangles}

The algorithm of Alt \& Scharf~\cite{as-cdaaa-13}
is an $O(n\log n)$ algorithm that computes
the depth of a collection of axis-aligned rectangles
in general position, using $O(n)$ space.
We give here a minor modification
which does not assume general position,
i.e., rectangles may share boundaries.
Moreover, it can consider each of the four boundaries
of a rectangle as either open or closed.

Given a set of finite rectangles,
the set of all $x$-coordinates of the
vertical sides of the input rectangles is first sorted.
Let $x_1,x_2,...,x_m$, $m<2n$ be the sorted set of $x$-coordinates.
The ordered set of intervals $\cal I$, is defined as follows;
For $i\in{1,2,...,m-1}$,
the $2(i-1)$th and $2(i-1)+1$th intervals in the set $\cal I$
are $[x_{i},x_{i}]$ and $(x_{i},x_{i+1})$, respectively.
The last interval is $[x_m,x_m]$.
A balanced binary tree $T$ is then constructed, 
holding all intervals in $\cal I$ in its leaves,
according to their order in $\cal I$.
An internal node represents the union of
the intervals of its two children,
which is a continuous interval.
In addition, each internal node $v$ stores in a variable $v.x$
the $x$-value of the merge point between the intervals
of its two children. Since we extended the algorithm
to support both open or closed boundaries, internal nodes
also maintain, a flag indicating whether the merge point
is to the left or to the right of the $x$-value.

According to the algorithm in~\cite{as-cdaaa-13},
a sweep is performed using a horizontal line from
$y=\infty$ to $y=-\infty$.
The sweep-line events occur when a rectangle starts or ends,
i.e., when top or bottom boundary of a rectangle is reached.
Since the rectangles are not in general position,
several events may share the same $y$-coordinate.
In such a case, the order of event processing in 
each $y$-coordinate is as follows:
\begin{enumerate}
  \item Closing rectangle with open bottom boundary events
  \item Opening rectangle with closed top boundary events
  \item Closing rectangle with closed bottom boundary events
  \item Opening rectangle with open top boundary events
\end{enumerate}
The order of event processing within each of these
four groups in a specific $y$-coordinate is not important.

\ignore{
The algorithm in~\cite{as-cdaaa-13} requires
each internal node to maintain in a variable $x$
the $x$-value of the merge point between the intervals
of its two children. Since we extended the algorithm
to support both open or closed boundaries, internal nodes
need to maintain, in addition to the $x$-value of the
intervals' merge point, a flag indicating whether the merge point
is to the left or to the right of the $x$-value.
}

The basic idea of the algorithm is that
each sweep event updates the appropriate
leaves of the tree $T$
(update the relevant leaves, spanning the covered
intervals, with the current event).
Therefore, each leaf holds a counter $c$ for the
number of covering rectangles in the current position
of the horizontal sweep line.
In addition, each leaf maintains in
a variable $c_m$ the maximal number of
covering rectangles for this leaf seen so far.
Clearly, the maximal coverage of an interval
is the maximal $c_m$ of all leaves.
The problem with this na\"{i}ve approach
is that one such update can already
take $O(n)$ time.
Therefore, the key idea of~\cite{as-cdaaa-13}
is that when updating an event of a rectangle
whose $x$-range is $(a,b)$,
one should follow only two paths;
the path to $a$ and the path to $b$.
The nodes on the path should hold the
information of how to update the
interval spanned by their children.
In the end of the update the union of
intervals spanned by the updated nodes
(internal nodes and only 2 leaves) is $(a,b)$.

In order to hold the information in the internal nodes
each internal node holds the following variables:
\begin{itemize}
\item[$l$] 
The additive update of rectangles that were opened or closed
and cover the interval spanned by the left child of $v$
since the last traversal of that child
\item[$r$] 
The additive update of rectangles that were opened or closed
and cover the interval spanned by the right child of $v$
since the last traversal of that child
\item[$l_m$] 
A counter which is used to count the maximum of the additive update
for the left child since the last traversal of that child
\item[$r_m$] 
A counter which is used to count the maximum of the additive update
for the right child since the last traversal of that child
\end{itemize}

\noindent
A leaf, on the other hand, holds two variables:
\begin{itemize}
\item[$c$] 
The coverage of the associated interval during the sweep until the last traversal on the leaf
\item[$c_m$] 
The maximum coverage of the associated  
interval during the sweep until the last traversal on the leaf
\end{itemize}

\noindent
In realation to these values we define the following functions:
\[
 t(v) =  \left\{ \begin{array}{ll}
        u.l + t(u) & \mbox{if $v$ is the left child of $u$} \\
        u.r + t(u) & \mbox{if $v$ is the left child of $u$} \\
        0 & \mbox{if $v$ is the root} \\
        \end{array} \right. 
\]

\[
 t_m(v) =  \left\{ \begin{array}{ll}
        max(u.l_m, u.l + t_m(u)) & \mbox{if $v$ is the left child of $u$} \\
        max(u.r_m, u.r + t_m(u)) & \mbox{if $v$ is the left child of $u$} \\
        0 & \mbox{if $v$ is the root} \\
        \end{array} \right.
\]

\pagebreak
\noindent 
At any point of the sweep the following two invariants hold for every leaf $\ell$ 
and its associated interval $I$: 
\begin{itemize}
    \item  
    The current coverage of $I$ is: $\ell.c + t(\ell)$    
    \item 
    The maximal coverage of $I$ that was seen so far is: $\max(\ell.c_m, \ell.c + t_m(\ell)$
\end{itemize}


\noindent
Updating the structure with an event is done as follows:
Let $I$ be the $x$-interval spanned by the processed rectangle creating the event.
Depending on whether the rectangle starts or ends,
we set a variable $d=1$ or $d=-1$, respectively. 
We follow the two search paths to the leftmost leaf and the rightmost leaf that are covered by $I$.
In the beginning the two paths are joined until they split, for every node $w$ on this path 
(including the split node) we can ignore $d$ and simply 
update the tuple $(w.l,w.r,w.l_m,w.r_m)$ using $t(w)$ and $t_m(w)$
according to the invariants stated above. Note that this process needs to clear the 
corresponding values in the parent node as otherwise the invariants would be violated.%
\footnote{Please note that using $t(w)$ and $t_m(w)$ here
takes constant time since we only need to access the parent 
node as all previous nodes on the path towards the root are already processed.} 
After the split the paths are processed separately, we 
discuss here the left path, the behavior for the right path is symmetric. 
Let $v$ be a node on the left path. 
As long as $v$ is not a leaf we update $(v.l,v.r,v.l_m,v.r_m)$ as usual.
However, if the path continues to the left we also have to incorporate $d$ into 
$v.r$ and $v.r_m$ as the subtree to the right is covered by $I$.
If $v$ is a leaf we simply update $v.c$ and $v.c_m$ using $t(v), t_m(v)$ and $d$.
A more detailed description (including pseudo code) can be found in~\cite{as-cdaaa-13}. 
In total, this process takes $O(\log n)$ time.

Finally, in order to find the maximal number of rectangles covering
an interval a final propagation from root to leaves is needed, such that
all $l, r, l_m, r_m$ values of internal nodes are cleared. This is done using one
traversal on $T$.
Now, the maximal number of rectangles covering an interval
is the maximal $c_m$ of all leaves of $T$.

Clearly, the running time of the algorithm is $O(n\log n)$,
since constructing the tree and sorting the $y$-events takes
$O(n\log n)$ time. Updating each of the $2n$ $y$-events takes $O(\log n)$
time, and the final propagation of values to the leaves takes $O(n)$ time.
The algorithm uses $O(n)$ space.


We remark that the above algorithm is not optimal in
memory usage. A more efficient variant which stores less variables
in the nodes of the tree can be easily implemented.

\subsecspacea\subsection{A Depth Preserving Reduction}\subsecspaceb
\label{ssec:reduction}
%
Let $T$ be a collection of open trapezoids with $y$-axis parallel bases
with the following property:
if the top or bottom curves of two different trapezoids intersect
then the two curves overlap completely in their joined $x$-range.
Let ${\cal A}(T)$ denote the arrangement of the trapezoids in $T$.
We describe a reduction from $T$ to $R$, where $R$ is a collection of
axis-parallel rectangles, such that the maximum depth
in ${\cal A}(R)$ equals to the maximum depth in ${\cal A}(T)$.

\noindent
The following observation is by~\cite{GY-TSR-80}:\\

\noindent
{\bf Observation 2.}
{\em
Let $S$ be a set of interior disjoint $x$-monotone curves.
There exists a partial order on $S$, such that if the
curves are moved one at a time to the direction of
$y=-\infty$ according to this order,
then each curve can be moved without (interior) intersecting
any of the remaining curves.
This order can be extended to a total order.
}

\begin{definition}
\label{def:ord}
Let $C$ be a set of interior disjoint $x$-monotone curves.
For two such curves $cv_i, cv_j \in C$,
let the open interval $(a,b)$ be the $x$-range of $cv_i$
and the open interval $(c,d)$ be the $x$-range of $cv_j$.
We define the total order $\prec$ as follows:\\
If $x\text{-range}(cv_i) \bigcap x\text{-range}(cv_j)=\emptyset$ then:\\
\mbox{\ \ } $cv_i \prec cv_j \Leftrightarrow b \leq c$\\
If $x\text{-range}(cv_i) \bigcap x\text{-range}(cv_j) \neq\emptyset$ then:\\
\mbox{\ \ } $cv_i \prec cv_j \Leftrightarrow
 cv_i(x) < cv_j(x) \text{ for some } x \in x\text{-range}(cv_i) \bigcap x\text{-range}(cv_j)$.
\end{definition}

\begin{definition}
Let $Ord$ denote a function $Ord: C\rightarrow\{1,...,n\}$
returning the order of a given $x$-monotone curve $cv \in C$
when sorting $C$ according to~$\prec$.
\end{definition}

\begin{definition}
\label{def:reduction}
We define a reduction from $T$ to $R$ as follows;
Every trapezoid $t \in T$ is reduced to a rectangle $r \in R$,
such that:
\begin{itemize}
  \item $t$ and $r$ have the same $x$-range,\\
  i.e. $(left(t) = left(r))$ and $(right(t) = right(r))$,
  where $left$ and $right$ denote the left $x$-value
  and the right $x$-value of $t$ (or $r$), respectively.
  \item The top and bottom edges of $r$ (accessible by $top$
  and $bottom$ methods) lie on $y=Ord(top(t))$ and
  $y=Ord(bottom(t))$, respectively.
\end{itemize}
\end{definition}

As shown in~\cite{CG-alg-app}, one can partition the plane into
vertical slabs by passing a vertical line through every endpoint
of the subdivision, and then partition each slab into regions by
intersecting it with all possible curves in the subdivision.
This defines a decomposition of the plane into at most $2(n+1)^2$
regions.

\begin{lemma}
\label{lemma:regions_bijection}
Let $Regions(arr)$ denote the collection of regions
of an arrangement $arr$, as defined above.
For any region $a_t \in Regions({\cal A}(T))$
let $a_r \in Regions({\cal A}(R))$ be the matching
rectangular region to $a_t$.
The collection $Regions({\cal A}(R))$
of all such rectangular regions spans the plane.
\end{lemma}

\begin{proof}
Trivial. The slabs remain the same and
within each slab the rectangular regions remain adjacent.
\qed
\end{proof}

\ignore{
\begin{proof}
Definition~\ref{def:reduction} implies that
$x\text{-range}(a_t) = x\text{-range}(a_r)$.
Therefore, the collection of rectangular regions
$Regions({\cal A}(R))$ spans the plane horizontally.
We now need to verify that it also spans the plane vertically.
We'll show that for two adjacent regions
$a_{t_{i}}, a_{t_{j}} \in Regions({\cal A}(T))$
for which $bottom(a_{t_{i}}) = top(a_{t_{j}})$,
then  $bottom(a_{r_{i}}) = top(a_{r_{j}})$ is true,
where $a_{r_{i}}, a_{r_{j}} \in Regions({\cal A}(R))$
are the matching rectangular regions
to regions $a_{t_{i}}, a_{t_{j}}$, respectively.
Since $bottom(a_{t_{i}}) = top(a_{t_{j}})$,
then $Ord(bottom(a_{t_{i}})) = Ord(top(a_{t_{j}}))$.
Therefore, according to Definition~\ref{def:reduction},
$bottom(a_{r_{i}})$ and $top(a_{r_{j}})$ lie
on the same horizontal line.
Since the $x$-range remains unchanged and the regions
are rectangular,
we get that $bottom(a_{r_{i}}) = top(a_{r_{j}})$.
In other words, since the collection $Regions({\cal A}(t))$
spans the plane
we conclude that the collection $Regions({\cal A}(R))$
spans the plane as well.
\qed
\end{proof}
}

\begin{lemma}
\label{lemma:reduction_from_trpz_to_rect}
Let $a_t \in Regions({\cal A}(T))$ be a region,
whose matching region is $a_r \in Regions({\cal A}(R))$.
The number of rectangles in $R$ that cover $a_r$
is at least the number of trapezoids in $T$ that cover $a_t$.
In other words,
for every $t \in T$ that covers $a_t$
its matching rectangle $r \in R$ covers $a_r$.
\end{lemma}

\begin{proof}
Let $\{t_1,t_2,...,t_m\}\subseteq T$
be the set of trapezoids, ordered by creation time,
such that for every $i\in\{1,...,m\}$, $t_i$ covers $a_t$.
Let $\{r_1,r_2,...,r_m\}\subseteq R$ be the set of matching
rectangles, such that $r_i$ matches $t_i$ for $i \in \{1...m\}$.
For any $t_i$, since $t_i$ covers $a_t$
we get that $x\text{-range}(a_t)\subseteq x\text{-range}(t_i)$.
By Definition~\ref{def:reduction} the
$x$-ranges remain the same after the reduction,
and therefore $x\text{-range}(a_r) \subseteq x\text{-range}(r_i)$.
Since $t_i$ covers $a_t$ then we also get that
in the shared $x$-range $top(t_i)$ is above or on $top(a_t)$
and $bottom(t_i)$ is below or on $bottom(a_t)$.
According to Definition~\ref{def:reduction},
it immediately follows that
$Ord(top(t_i)) \geq Ord(top(a_t))$.
In other words, $top(r_i)$ is above or on $top(a_r)$.
Similarly, $bottom(r_i)$ is below or on $bottom(a_r)$.
We conclude that $r_i$ covers $a_r$.
\qed
\end{proof}

\begin{lemma}
\label{lemma:reduction_from_rect_to_trpz}
Let $a_r\in Regions({\cal A}(R))$ be a rectangular region,
whose matching region is $a_t \in Regions({\cal A}(T))$.
The number of trapezoids in $T$ that cover $a_t$
is at least the number of rectangles in $R$ that cover $a_r$.
In other words,
for every $r \in R$ that covers $a_r$
its matching trapezoid $t \in T$ covers $a_t$.
\end{lemma}

\begin{proof}
Let $\{r_1,r_2,...,r_m\}\subseteq R$
be the set of rectangles,
such that for every $i\in\{1,...,m\}$, $r_i$ covers $a_r$.
Let $\{t_1,t_2,...,t_m\}\subseteq T$ be the set of matching
trapezoids, such that $t_i$ matches $r_i$ for $i \in \{1...m\}$.
Proving that for any $i \in \{1...m\}$, $t_i$ covers $a_t$,
is done symmetrically to the proof of Lemma~\ref{lemma:reduction_from_trpz_to_rect}.
\qed
\end{proof}

\ignore{
\begin{proof}
Let $\{r_1,r_2,...,r_m\}\subseteq R$
be the set of rectangles,
such that for every $i\in\{1,...,m\}$, $r_i$ covers $a_r$.
Let $\{t_1,t_2,...,t_m\}\subseteq T$ be the set of matching
trapezoids, such that $t_i$ matches $r_i$ for $i \in \{1...m\}$.
for any $r_i$, since $r_i$ covers $a_r$
we get that $x\text{-range}(a_r)\subseteq x\text{-range}(r_i)$.
By Definition~\ref{def:reduction} we know that both
$x\text{-range}(r_i) = x\text{-range}(t_i)$ and
$x\text{-range}(a_r) = x\text{-range}(a_t)$ are true.
Therefore, $x\text{-range}(a_t) \subseteq x\text{-range}(t_i)$
is true.
Since $r_i$ covers $a_r$ then we also get that
in the shared $x$-range $top(r_i)$ is above or on $top(r_t)$
and $bottom(r_i)$ is below or on $bottom(a_r)$.
We will show that in the shared $x$-range
of $t_i$ and $a_t$, $top(t_i)$ is above or on $top(a_t)$ as well.
According to Definition~\ref{def:reduction},
$Ord(top(t_i)) \geq Ord(top(a_t))$.
In other words, $top(t_i)$ is above or on $top(a_t)$.
Similarly, $bottom(t_i)$ is below or on $bottom(a_t)$.
We conclude that $t_i$ covers $a_t$.
\end{proof}
}

Combining Lemma~\ref{lemma:reduction_from_trpz_to_rect} and
Lemma~\ref{lemma:reduction_from_rect_to_trpz}
we conclude that the number of trapezoids in $T$
that cover a region $a_t$ equals to the number
of rectangles in $R$ that cover $a_r$, which is
the matching region to $a_t$.
The covering rectangles are the reduced trapezoids
in the set of trapezoids covering $a_t$.
Since both $Regions({\cal A}(T))$
and $Regions({\cal A}(R))$ span the plane
(Lemma~\ref{lemma:regions_bijection}),
we get the following theorem.

\begin{theorem}
\label{theorem:reduction_correctness}
Let $T$ be a collection of open trapezoids with the following properties:
their bases are $y$-axis parallel (vertical walls) and if
the top or bottom curves of two different trapezoids intersect
then the two curves overlap completely in their joined $x$-range.
Let ${\cal A}(T)$ denote the arrangement of the trapezoids in~$T$.
$T$ can be reduced to a collection of open axis-parallel rectangles~$R$,
such that the maximum depth in ${\cal A}(R)$ equals
to the maximum depth in ${\cal A}(T)$.
\end{theorem}

\subsecspacea\subsection{Modification of Alt \& Scharf}\subsecspaceb
\label{ssub:alg_arr_depth_general}
Based on the correctness of the reduction described in Subsection~\ref{ssec:reduction}
we can extended the basic algorithm presented in Subsection~\ref{ssub:alg_arr_depth_rectangles}
to support not only collections of axis-aligned rectangles but
also collections of open trapezoids with $y$-axis parallel bases 
(vertical walls) and
non-intersecting top and bottom boundaries (if they intersect
then they overlap completely in their joined $x$-range).
The only part of the basic algorithm that should change
is the top-to-bottom sweep.
Therefore, the simple predicate in \cite{as-cdaaa-13} that is used 
for sorting the $y$-events should be replaced with a new predicate that
compares according to the reverse order of $\prec$, as given in Definition~\ref{def:ord}.

Please note that for simplicity
we assumed that no two distinct endpoints
in the original subdivision have
the same $x$-value. 
However, if this is not the case,
a lexicographical compare can be used on the endpoints
of the curves in order to define the order 
of the induced vertical walls.

\ignore{

\newcommand{\vrcm}{v . r . c_m}
\newcommand{\vlcm}{v . l . c_m}
\newcommand{\vrc}{v . r . c}
\newcommand{\vlc}{v . l . c}
\newcommand{\vc}{v . c }
\newcommand{\vcm}{v . c_m}
\newcommand{\vx}{v .x}
\newcommand{\vl}{v .l}
\newcommand{\vr}{v .r}

\begin{algorithm}[H]
\eIf{$a \leq \vx $}
{
  $\vrcm = \max\{\vrcm,\vrc+\vcm,\vrc+\vc+d\}$;\\
  $\vrc  = \vrc+\vc+d$;\\
  $\vlcm = \max\{\vlcm,\vlc+\vcm\}$;\\
  $\vlc  = \vlc+vc$;\\
  $\vc = vcm = 0$;\\
  $LS(a,\vl,d)$;\\
}{}
\caption{SearchLeft(v,a,d)}
\end{algorithm}

} 

\end{document}